%% file: main.tex
\documentclass[submission,copyright,creativecommons]{eptcs}
\providecommand{\event}{ICLP 2025} 

\usepackage{iftex}

\ifpdf
  \usepackage{underscore}         
  \usepackage[T1]{fontenc}        
\else
  \usepackage{breakurl}           
\fi

\usepackage{mathptmx}
\usepackage{amsmath}

\input{sections/macro}

\providecommand{\keywords}[1]
{
	\small	
	\textbf{\textit{Keywords:}} #1
}

\title{On the Trap Space Semantics of Normal Logic Programs}
\author{Van-Giang Trinh \qquad\qquad Sylvain Soliman \qquad\qquad Fran\c{c}ois Fages
\institute{Inria Saclay, EP Lifeware, Palaiseau, France}
\email{\quad van-giang.trinh@inria.fr \quad\qquad Sylvain.Soliman@inria.fr \quad\qquad francois.fages@inria.fr}
\and Belaid Benhamou
\institute{LIRICA team, LIS, Aix-Marseille University, Marseille, France}
\email{\quad belaid.benhamou@lis-lab.fr}
}
\def\titlerunning{On the Trap Space Semantics of Normal Logic Programs}
\def\authorrunning{V.G. Trinh, S. Soliman, F. Fages \& B. Benhamou}
\begin{document}
\maketitle

\begin{abstract}
The logical semantics of normal logic programs has traditionally been based on the notions of Clark's completion and two-valued or three-valued canonical models, including supported, stable, regular, and well-founded models.
Two-valued interpretations can also be seen as states evolving under a program's update operator, producing a transition graph whose fixed points and cycles capture stable and oscillatory behaviors, respectively.
We refer to this view as \emph{dynamical semantics} since it characterizes the program's meaning in terms of state‐space trajectories, as first introduced in the stable (supported) class semantics.
Recently, we have established a formal connection between \DLN programs (i.e., normal logic programs without function symbols) and Boolean networks, leading to the introduction
of the \emph{trap space} concept for \DLN programs.
In this paper, we generalize the trap space concept to arbitrary normal logic programs, introducing \emph{trap space semantics} as a new approach to their interpretation.
This new semantics admits both model-theoretic and dynamical characterizations, providing a comprehensive approach to understanding program behavior.
We establish the foundational properties of the trap space semantics and systematically relate it to the established model-theoretic semantics, including the stable (supported), stable (supported) partial, regular, and L-stable model semantics, as well as to the dynamical stable (supported) class semantics.
Our results demonstrate that the trap space semantics offers a unified and precise framework for proving the existence of supported classes, strict stable (supported) classes, and regular models, in addition to uncovering and formalizing deeper relationships among the existing semantics of normal logic programs.

\end{abstract}

\input{sections/introduction}
\input{sections/background}
\input{sections/ts-def-pro}
\input{sections/relationship}
\input{sections/consistency}
\input{sections/conclusion}

%

\bibliographystyle{eptcs}
\bibliography{ref}

\newpage
\appendix

\input{sections/detailed-proofs}
\input{sections/characterization-st-su-class}
\input{sections/existence-L-stable}

\end{document}

%% file: sections/macro.tex
\usepackage{amsthm}

\usepackage{subcaption}
\usepackage{amssymb}
\usepackage[inline]{enumitem}

\usepackage{stmaryrd} 

\usepackage{hyperref}
\usepackage{xcolor}

\usepackage{xspace}

\usepackage{tikz}
\usetikzlibrary{arrows,shapes,automata,petri,positioning}
\usetikzlibrary{arrows.meta, positioning}

\tikzset{
	place/.style={
		circle,
		thick,
		draw=blue!75,
		fill=blue!20,
		minimum size=6mm,
	},
	transitionH/.style={
		rectangle,
		thick,
		fill=black,
		minimum width=8mm,
		inner ysep=2pt
	},
	transitionV/.style={
		rectangle,
		thick,
		fill=black,
		minimum height=8mm,
		inner xsep=2pt
	}
}

\usepackage{forest}
\usepackage{tikz-qtree}

\newcommand{\dng}[1]{{{\sim}#1}}
\newcommand{\gr}[1]{\mathsf{gr}({#1})}

\newcommand{\head}[1]{\mathsf{h}({#1})}
\newcommand{\body}[1]{\mathsf{b}({#1})}
\newcommand{\bodyf}[1]{\mathsf{bf}({#1})}
\newcommand{\pbody}[1]{\mathsf{b}^+({#1})}
\newcommand{\nbody}[1]{\mathsf{b}^-({#1})}

\newcommand{\comp}[1]{\mathsf{comp}({#1})}

\newcommand{\lfp}[1]{\mathsf{lfp}({#1})}
\newcommand{\tgsp}[1]{\mathsf{tg}_{sp}({#1})}
\newcommand{\tgst}[1]{\mathsf{tg}_{st}({#1})}
\newcommand{\rhs}[1]{\mathsf{rhs}_{P}({#1})}
\newcommand{\hb}[1]{\mathsf{HB}_{#1}}

\newcommand{\Top}[1]{T_{P}({#1})}
\newcommand{\Fop}[1]{F_{P}({#1})}

\newcommand{\cset}[1]{\mathcal{S}({#1})}

\newcommand{\stms}[1]{\mathsf{StM}({#1})}
\newcommand{\sums}[1]{\mathsf{SuM}({#1})}
\newcommand{\stpms}[1]{\mathsf{StPM}({#1})}
\newcommand{\supms}[1]{\mathsf{SuPM}({#1})}
\newcommand{\regms}[1]{\mathsf{RegM}({#1})}
\newcommand{\suts}[1]{\mathsf{SuTS}({#1})}
\newcommand{\stts}[1]{\mathsf{StTS}({#1})}

\newcommand{\twod}[1]{\mathbb{B}}
\newcommand{\threed}[1]{\mathbb{B}_{\star}}

\newcommand{\tval}[0]{1}
\newcommand{\fval}[0]{0}
\newcommand{\uval}[0]{\star}

\newcommand\DLN{Datalog$^\neg$\xspace}

\newcommand{\pname}[0]{\DLN program\xspace}
\newcommand{\pnames}[0]{\DLN programs\xspace}

\newcommand{\translfp}[2]{\Pi_{{#1}}{#2}}

\newcommand{\acnlp}[1]{\text{NLP}}
\newcommand{\acdg}[1]{\text{DG}}
\newcommand{\acstpm}[1]{\text{StPM}}
\newcommand{\acsupm}[1]{\text{SuPM}}
\newcommand{\acstm}[1]{\text{StM}}
\newcommand{\acsum}[1]{\text{SuM}}
\newcommand{\acregm}[1]{\text{RegM}}

\newcommand{\acbn}[1]{\text{BN}}
\newcommand{\acig}[1]{\text{IG}}
\newcommand{\acstg}[1]{\text{STG}}

\newcommand{\nlp}[0]{\text{NLP}\xspace}
\newcommand{\nlps}[0]{\text{NLPs}\xspace}
\newcommand{\af}[0]{\text{AF}\xspace}
\newcommand{\afs}[0]{\text{AFs}\xspace}

\newcommand{\ifftext}{\text{iff}\xspace}
\newcommand{\wrttext}{\text{w.r.t.}\xspace}

\usepackage[createShortEnv,disablePatchSection,
  commandRef=Cref]{proof-at-the-end}

\newtheorem{theorem}{Theorem}[section]
\newtheorem{lemma}{Lemma}[section]

\newtheorem{proposition}{Proposition}[section]

\newtheorem{definition}{Definition}[section]
\newtheorem{example}{Example}[section]

\newtheorem{conjecture}{Conjecture}[section]

%% file: sections/introduction.tex
\section{Introduction}\label{sec:introduction}

The meaning of a Normal Logic Program (\nlp)---its declarative semantics---is typically defined by a class of models of a logical formula associated with the program, such as the Herbrand models of the \emph{Clark's completion} of the program~\cite{Clark1977}.
This was later shown to be equivalent to the \emph{supported model semantics}~\cite{ABW1988}.
For \nlps with negation, the Clark's completion may have no model, or several minimal models; therefore,~there is generally no unique least Herbrand model.
In this setting, the \emph{stable model semantics} was introduced by~\cite{GL1988} to provide a canonical notion of two-valued models, defined by an elegant fixed‑point definition based on the least model of a \emph{reduct} transformation.
Stable models were later characterized as \emph{well-supported models}~\cite{Fages1991}, i.e.,~models with circuit-free finite justifications---which led to graphical conditions for the existence of stable models~\cite{Fages1994,AB1991}, and equivalence results with other forms of non-monotonic logics~\cite{YY1995,MS1992,Pearce2006,BS2012}.

A three-valued logic approach was introduced, where the third value denotes ignorance~\cite{GRS1991,Przymusinski1990,SZ1997}.
Subsequent research expanded the semantic landscape of \nlps through additional three-valued logic transformations~\cite{AD1995}. 
The \emph{stable (supported) partial model semantics}—also known as the three-valued stable (supported) model semantics—extends the two-valued stable (supported) model semantics~\cite{P1994,ELS1997,Przymusinski1990}. 
Under this framework, all models are minimal, and every \nlp possesses at least one stable (supported) partial model. 
Notably, two specific types of stable partial semantics are distinguished by their informational content: the \emph{well-founded partial model semantics}~\cite{GRS1988,Przymusinski1990}, representing the unique model with minimal information, and the \emph{regular model semantics}~\cite{YY1994}, encompassing models with maximal information.
It is worth noting that these three-valued semantics (e.g., stable partial models, regular models) in normal logic programs have been shown to correspond to several extension-based semantics (e.g., complete extensions, preferred extensions, respectively) in Dung's frameworks~\cite{Dung1995,WCG2009} and assumption-based argumentation~\cite{CS2017}, which are two prominent formalisms in symbolic AI.

Alternatively, an \nlp can be viewed as a state-transition system in which each rule in the program represents how a component (here an atom) is affected by other components~\cite{BS1992,BDH1997}.
The dynamical behaviors of the program, including its oscillating behavior~\cite{IS2012}, are characterized by a directed graph called the (stable or supported) \emph{transition graph} whose vertices are two-valued interpretations and whose arcs are transitions between interpretations.
This led to the \emph{stable class semantics}~\cite{BS1992} and \emph{supported class semantics}~\cite{IS2012}.
If a stable (resp.\ supported) class is of size one, it coincides with a stable (resp.\ supported) model and exhibits a stable behavior.
Otherwise, it exhibits an oscillating behavior.

The aforementioned semantics capture various aspects of program behavior.
Despite their success, they often diverge in interpretive principles, and many lack a unifying framework that integrates both model-theoretic and dynamical behaviors coherently.
Recently, we have established a formal connection between \pnames (i.e., \nlps without function symbols---thus the Herbrand base is finite) and Boolean networks, leading to the introduction of the concept of \emph{trap space} for \pnames, inspired by the notion of trap space in Boolean networks~\cite{KBS2015,TBS2023}.
This concept offers a bridge between model-theoretic and dynamical aspects, but applies only to finite structures~\cite{TBSF2025}.

In this paper, we generalize the trap space concept to arbitrary \nlps, introducing \emph{trap space semantics} as a new approach to their interpretation.
This new semantics admits both model-theoretic and dynamical characterizations, providing a comprehensive approach to understanding program behavior.
We establish foundational properties of the trap space semantics and systematically relate it to established model-theoretic and dynamical semantics---including the stable (supported), stable (supported) partial, regular, and L-stable model semantics---as well as to dynamical semantics such as the stable (supported) class semantics.
Our results demonstrate that the trap space semantics offers a unified framework (along with a proper use of set-theoretic axioms) for proving the existence of supported classes, strict supported classes, strict stable classes, or regular models---proofs that have been unfortunately either missing or imprecise in prior literature~\cite{BS1992,YY1994,IS2012}.
Furthermore, by leveraging the concept of stable or supported trap set underlying the trap space semantics, we correct a flaw in the existing alternative characterization of strict stable or supported classes established in~\cite{IS2012}.
Finally, the trap space semantics emerges as a powerful tool for uncovering and formalizing deeper relationships among various existing semantics of \nlps, for instance, the stable (resp.\ supported) class semantics vs. the stable (resp.\ supported) partial model semantics, and the regular model semantics vs. the stable class semantics and the L-stable model semantics.

The rest of the paper is organized as follows:
in Section~\ref{sec:preliminaries}, we review the necessary background on \nlps and set theory.
Section~\ref{sec:TS-defs-properties} presents the formal definitions of the trap space semantics, as well as its fundamental properties. 
Section~\ref{sec:TS-relationships} shows the relationships between the trap space semantics and other semantics.
In Section~\ref{sec:semantics-consistency}, we present the existence results.
Section~\ref{sec:conclusion} offers concluding remarks and future research directions.

%% file: sections/background.tex
\section{Preliminaries}\label{sec:preliminaries}

\subsection{Normal Logic Programs}\label{subsec:preliminaries-NLP}

We assume that the reader is familiar with the logic programming theory and stable model semantics.
Unless specifically stated, \nlp means normal logic program.
In addition, we consider the Boolean domain \(\twod{} = \{\fval, \tval\}\), the three-valued domain \(\threed{} = \{\fval, \tval, \uval\}\).
The logical connectives used in this paper are \(\land\) (conjunction), \(\lor\) (disjunction), \(\neg\) (negation), and \(\leftrightarrow\) (equivalence).

We consider a first-order language built over an infinite alphabet of variables,
and finite alphabets of constant, function, and predicate symbols.
The set of first-order \emph{terms} is the least set containing variables and constants, and closed by application of function symbols.
An \emph{atom} is a formula of the form \(p(t_1, \dots, t_k)\) where \(p\) is a predicate symbol and each \(t_i\) is a term.
A \emph{normal logic program} \(P\) is a \emph{finite} set of \emph{rules} of the form \(p \gets p_1, \dots, p_m, \dng{p_{m + 1}}, \dots, \dng{p_{k}}\)
where \(p\) and \(p_i\) are atoms, \(k \geq m \geq 0\), and \(\sim\) denotes default negation.
A fact is a rule with \(k = 0\).
For any rule \(r\) of the above form, \(\head{r} = p\) is called the \emph{head} of \(r\), \(\pbody{r} = \{p_1, \dots, p_m\}\) is called the \emph{positive body} of \(r\), \(\nbody{r} = \{p_{m + 1}, \dots, p_{k}\}\) is called the \emph{negative body} of \(r\), and \(\body{r} = \pbody{r} \cup \nbody{r}\) is called the \emph{body} of \(r\). 
For convenience, we denote \(\bodyf{r} = \bigwedge_{v \in \pbody{r}}v \wedge \bigwedge_{v \in \nbody{r}}\neg v\) as the \emph{body formula} of \(r\); if \(k = 0\), then \(\bodyf{r} = \tval\).
If \(\nbody{r} = \emptyset\) forall \(r \in P\), then \(P\) is called a \emph{positive} \nlp.
If \(\pbody{r} = \emptyset\) forall \(r \in P\), then \(P\) is called a \emph{negative} \nlp.
The program \(P\) is called \emph{internal-variable-free} if for any rule in \(P\), the variables in the body also appear in the head~\cite{Sato1990}.

A term, an atom, or an \nlp is \emph{ground} if it contains no variables.
The \emph{Herbrand base} of an \nlp \(P\) (denoted by \(\hb{P}\)) is the set of ground atoms formed over the alphabet of \(P\).
It is finite in the absence of function symbols, which is the case of \emph{\DLN} programs~\cite{CGT1990}.
The \emph{ground instantiation} of an \nlp \(P\) (denoted by \(\gr{P}\)) is the set of the ground instances of all rules in \(P\).
An \nlp \(P\) is called \emph{uni-rule} if every atom in \(\hb{P}\) appears in the head of at most one rule in \(\gr{P}\)~\cite{SS1997}.

We shall use the \emph{least fixpoint} transformation of \nlps~\cite{DK1989} to prove many new results in the next sections.
To be self-contained, we briefly recall the definition of the \emph{least fixpoint} of an \nlp \(P\) as follows.
The reader can refer~\cite{DK1989} for more details.
Let \(r\) be the rule \(p \leftarrow \dng{p_1}, \dots, \dng{p_k}, q_1, \dots, q_j\) in \(\gr{P}\) and let \(r_i\) be the rules \(q_i \leftarrow \dng{q^1_i}, \dots, \dng{q^{l_i}_i}\) in \(\gr{P}\) where \(1 \leq i \leq j\) and \(l_i \geq 0\).
Then \(\translfp{r}{(\{r_1, \dots, r_j\})}\) is the following rule \[p \leftarrow \dng{p_1}, \dots, \dng{p_k}, \dng{q_1^1}, \dots, \dng{q_1^{l_1}}, \dots, \dng{q_j^1}, \dots, \dng{q_j^{l_j}},\] which means that each atom \(q_j\) in the positive body of \(r\) is substituted with the body of the rule \(r_j\).
Then \(\translfp{P}{}\) is the transformation on negative normal logic programs: \[\translfp{P}{(Q)} = \{\translfp{r}{(\{r_1, \dots, r_j\})} \mid r \in \gr{P}, r_i \in Q, 1 \leq i \leq j\}.\]
Let \(\lfp{P}_i = (\translfp{P}{(\emptyset}))^i = \translfp{P}{(\translfp{P}{(\dots \translfp{P}{(\emptyset)}))}}\), then \(\lfp{P} = \bigcup_{i \geq 1}\lfp{P}_i\) is the least fixpoint of the \nlp \(P\).


\subsubsection{Stable and Supported Partial Models}

A \emph{three-valued interpretation} \(I\) of an \nlp \(P\) is a mapping \(I \colon \hb{P} \to \threed{}\).
If \(I(a) \neq \uval\) for all \(a \in \hb{P}\), then \(I\) is a \emph{two-valued interpretation} of \(P\).
Usually, a two-valued interpretation is written as the set of atoms that are true.
A three-valued interpretation \(I\) characterizes  the set of two-valued interpretations denoted by \(\cset{I}\) as \(\cset{I} = \{J \mid J \in 2^{\hb{P}}, \forall a \in \hb{P}, I(a) \neq \uval \Rightarrow J(a) = I(a)\}\).
For example, if \(I = \{p = \tval, q = \fval, r = \uval\}\), then \(\cset{I} = \{\{p\}, \{p, r\}\}\).
We consider three orders on three-valued interpretations.
The truth order \(\leq_{t}\) is given by \(\fval <_{t} \uval <_{t} \tval\).
Then, \(I_1 \leq_t I_2\) \ifftext \(I_1(a) \leq_{t} I_2(a)\) forall \(a \in \hb{P}\).
The order \(\leq_{s}\) is given by \(\fval <_{s} \uval\), \(\tval <_{s} \uval\); no other relations hold.
Then, \(I_1 \leq_{s} I_2\) \ifftext \(I_1(a) \leq_{s} I_2(a)\) forall \(a \in \hb{P}\). 
In addition, \(I_1 \leq_{s} I_2\) \ifftext \(\cset{I_1} \subseteq \cset{I_2}\), i.e., \(\leq_{s}\) is identical to the subset partial order.
The information order \(\leq_{i}\) is given by \(\uval <_{i} \fval\), \(\uval <_{i} \tval\); no other relations hold.
Then, \(I_1 \leq_{i} I_2\) \ifftext \(I_1(a) \leq_{i} I_2(a)\) forall \(a \in \hb{P}\).
It is easy to see that \(\leq_i\) is the complement of \(\leq_s\), i.e., \(I_1 \leq_s I_2\) \ifftext \(I_2 \leq_i I_1\).

Let \(e\) be a propositional formula on \(\hb{P}\), using only the \(\neg\), \(\land\), and \(\lor\) connectives.
Then the valuation of \(e\) under a three-valued interpretation \(I\) (denoted by \(I(e)\)) following the three-valued logic is defined recursively as follows: \(I(e) = e\) if \(e \in \threed{}\); \(I(e) = I(a)\) if \(e = a, a \in \hb{P}\); \(I(e) = \neg I(e_1)\) if \(e = \neg e_1\); \(I(e) = \text{min}_{\leq_t}(I(e_1), I(e_2))\) if \(e = e_1 \land e_2\); and \(I(e) = \text{max}_{\leq_t}(I(e_1), I(e_2))\) if \(e = e_1 \lor e_2\)
where \(\neg \tval = \fval, \neg \fval = \tval, \neg \uval = \uval\), and \(\text{min}_{\leq_t}\) (resp.\ \(\text{max}_{\leq_t}\)) is the function to get the minimum (resp.\ maximum) value of two values w.r.t.\ the order \(\leq_t\).
We say three-valued interpretation \(I\) is a \emph{three-valued model} of an \nlp \(P\) if for each rule \(r \in \gr{P}\), \(I(\bodyf{r}) \leq_{t} I(\head{r})\).

Given a three-valued interpretation \(I\)  of an \nlp \(P\), the \emph{reduct} of \(P\) w.r.t. \(I\) (denoted by \(P^I\)) is defined as: remove any rule \(a \leftarrow a_1, \dots, a_m, \dng{b_1}, \dots, \dng{b_k} \in \gr{P}\) if \(I(b_i) = \tval\) for some \(1 \leq i \leq k\); afterwards, remove any occurrence of \(\dng{b_i}\) from \(\gr{P}\) such that \(I(b_i) = \fval\); then, replace any occurrence of \(\dng{b_i}\) left by a special atom \textbf{u} such that \(\textbf{u} \not \in \hb{P}\) and it always receives the value \(\uval\).
The \nlp \(P^I\) is positive and has a unique \(\leq_{t}\)-least three-valued model.
See~\cite{Przymusinski1990} for the method for computing this \(\leq_{t}\)-least model.
Then \(I\) is a \emph{stable partial model} of \(P\) if \(I\) is equal to the \(\leq_{t}\)-least three-valued model of \(P^I\).
A stable partial model \(I\) is an \emph{M-stable model} if it is \(\leq_s\)-minimal~\cite{SZ1997}.
In this paper, we define a \emph{regular model} as an M-stable model as they are equivalent~\cite{ELS1997}.
In the two-valued setting, \(P^I\) is identical to the reduct defined by~\cite{GL1988}; \(I\) is a \emph{stable model} of \(P\) if \(I\) is equal to the \(\leq_{t}\)-least two-valued model of \(P^I\).
It is easy to derive that a stable model is a two-valued regular model, as well as a two-valued stable partial model~\cite{AD1995}.
We define \(I^{\textbf{u}}\) as the set of ground atoms that are assigned to \(\uval\) in \(I\), i.e., \(I^{\textbf{u}} = \{v \mid v \in \hb{P}, I(v) = \uval\}\).
Then \(I\) is a \emph{L-stable model} of \(P\) if \(I\) is a stable partial model of \(P\) and there is no stable partial model \(J\) such that \(J^{\textbf{u}} \subset I^{\textbf{u}}\)~\cite{SZ1997}.

The (propositional) \emph{Clark's completion} of an \nlp \(P\) (denoted by \(\comp{P}\)) consists of the following equivalences: \(p \leftrightarrow \bigvee_{r \in \gr{P}, \head{r} = p}\bodyf{r}\), for each \(p \in \hb{P}\); if there is no rule whose head is \(p\), then the equivalence is \(p \leftrightarrow \fval\).
Let \(\rhs{a}\) denote the right-hand side of the equivalence of a ground atom \(a \in \hb{P}\) in \(\comp{P}\).
A three-valued interpretation \(I\) is a three-valued model of \(\comp{P}\) \ifftext for every \(a \in \hb{P}\), \(I(a) = I(\rhs{a})\).
In this work, we define a \emph{supported partial model} of \(P\) as a three-valued model of \(\comp{P}\), and a \emph{supported model} of \(P\) as a two-valued model of \(\comp{P}\).
Clearly, a supported model is a (two-valued) partial supported model.
It has been pointed out that a stable partial model (resp.\ stable model) is a supported partial model (resp.\ supported model), but the reverse may be not true~\cite{DHW2014}.

%

\subsubsection{Stable and Supported Classes}

Let \(P\) be an NLP and let \(I\) be any two-valued interpretation of \(P\).
We have that \(P^I\) is positive, and has a unique \(\leq_{t}\)-least two-valued model (say \(J\)).
We define the operator \(F_P\) as \(\Fop{I} = J\).
In contrast, we define the operator \(T_P\) as \(\Top{I} = J\) where \(J\) is a two-valued interpretation such that for every \(a \in \hb{P}\), \(J(a) = I(\rhs{a})\).

\begin{definition}\label{def:NLP-StC-SuC}
	A non-empty set \(S\) of two-valued interpretations is a \emph{stable class} of an \nlp \(P\) \ifftext it holds that \(S = \{\Fop{I} \mid I \in S\}\)~\cite{BS1992}.
	Similarly, \(S\) is a \emph{supported class} of an \nlp \(P\) \ifftext it holds that \(S = \{\Top{I} \mid I \in S\}\)~\cite{IS2012}.
	A stable (resp.\ supported) class \(S\) of \(P\) is \emph{strict} \ifftext no proper subset of \(S\) is a stable (resp.\ supported) class of \(P\).
\end{definition}

\begin{example}\label{exam:Datalog-St-Su-class}
	Consider the \nlp \(P = \{a \leftarrow b; b \leftarrow a\}\) where the symbol ``;'' is used to separate program rules.
	Since \(P\) is positive, for any two-valued interpretation \(I\), \(P^I = P\).
	It follows that \(\Fop{I} = \emptyset\) for all \(I\).
	It is easy to derive that \(P\) has only a unique stable class \(\{\emptyset\}\).
	On the other hand, we have that \(\Top{\emptyset} = \emptyset\), \(\Top{\{a\}} = \{b\}\), \(\Top{\{b\}} = \{a\}\), and \(\Top{\{a, b\}} = \{a, b\}\).
	It is easy to derive that \(P\) has three supported classes: \(\{\emptyset\}\), \(\{\{a\}, \{b\}\}\), and \(\{\{a, b\}\}\).
\end{example}


\subsubsection{Transition Graphs}

The \emph{stable} (resp.\ \emph{supported}) \emph{transition graph} of \(P\) is a (possibly infinite) directed graph (denoted by \(\tgst{P}\) (resp.\ \(\tgsp{P}\))) on the set of all possible two-valued interpretations of \(P\) such that \((I, J)\) is an arc of \(\tgst{P}\) (resp.\ \(\tgsp{P}\)) \ifftext \(J = \Fop{I}\) (resp.\ \(J = \Top{I}\)).
Given a directed graph \(G\), we use \(V(G)\) (resp.\ \(E(G)\)) to denote the set of vertices (resp.\ arcs) of \(G\).
For more details about theory of finite and infinite graphs, we refer the reader to~\cite{Konig1990}.

\begin{example}\label{exam:NLP-all}
	Consider the \nlp \(P = \{p(X) \leftarrow p(s(X))\}\).
	The Clark's completion of \(P\) includes \(p(0) \leftrightarrow p(s(0))\), \(p(s(0)) \leftrightarrow p(s(s(0)))\), \dots
	It is easy to verify that \(P\) has only three supported partial models: \(I_1\) where \(I_1(a) = \fval\) for every \(a \in \hb{P}\), \(I_2\) where \(I_2(a) = \tval\) for every \(a \in \hb{P}\), and \(I_3\) where \(I_3(a) = \uval\) for every \(a \in \hb{P}\).
	We have \(P^{I_1} = P^{I_2} = P^{I_3} = P\).
	It is easy to derive that \(I_1\) is the unique stable partial model of \(P\).
	It is also the unique regular model, stable model, and L-stable model of \(P\).
	The stable and supported transition graphs of \(P\) are shown in Figure~\ref{fig:exam-NLP-sttg-sptg}~(a) and Figure~\ref{fig:exam-NLP-sttg-sptg}~(b), respectively.
	The \nlp \(P\) has only one stable class \(\{I_1\}\), but three supported classes \(\{I_1\}\),  \(\{I_2\}\), and  \(\{I_1, I_2\}\).
	Then \(\{I_1\}\) is the unique strict stable class of \(P\), whereas \(\{I_1\}\) and  \(\{I_2\}\) are strict supported classes of \(P\).
\end{example}

\begin{figure}[!ht]
	\centering
	\begin{subfigure}{0.5\textwidth}
		\begin{tikzpicture}[
			node distance=0.6cm and 0.6cm,
			every node/.style={draw, ellipse, minimum width=0.5cm, minimum height=0.5cm, font=\small},
			arrow/.style={-{Latex[length=2mm]}, thick},
			dot/.style={circle, fill, inner sep=1pt},
			ellipsis/.style={draw=none, font=\large\bfseries}
			]
			
			\node (I1) at (0,0) {$I_1$};
			\node (a0) [above=of I1, xshift=0cm] {$\{a_0\}$};
			\node (a1) [right=of a0] {$\{a_1\}$};
			\node (a2) [right=of a1] {$\{a_2\}$};
			\node (dots1) [right=of a2, draw=none] {$\cdots$};
			\node (I2) [right=of I1, xshift=3cm] {$I_2$};
			
			\node (a0a1) [above=of a1] {$\{a0, a1\}$};
			\node (a1a2) [right=of a0a1] {$\{a1, a2\}$};
			\node (dots2) [right=of a1a2, draw=none] {$\cdots$};
			
			\node (dots3) [above=of a0a1, draw=none] {$\cdots$};
			\node (dots4) [above=of a1a2, draw=none] {$\cdots$};
			
			\draw[arrow] (a0a1) -- (I1);
			\draw[arrow] (a1a2) edge [bend left=10] (I1);
			\draw[arrow] (a2) -- (I1);
			\draw[arrow] (a1) -- (I1);
			\draw[arrow] (a0) -- (I1);
			
			\path [arrow] (I1) edge [out=180, in=135, loop] (I1);
			\path [arrow] (I2) edge [] (I1);
		\end{tikzpicture}
		\caption{}
	\end{subfigure}\begin{subfigure}{0.5\textwidth}
		\begin{tikzpicture}[
			node distance=0.6cm and 0.6cm,
			every node/.style={draw, ellipse, minimum width=0.5cm, minimum height=0.5cm, font=\small},
			arrow/.style={-{Latex[length=2mm]}, thick},
			dot/.style={circle, fill, inner sep=1pt},
			ellipsis/.style={draw=none, font=\large\bfseries}
			]
			
			\node (I1) at (0,0) {$I_1$};
			\node (a0) [above=of I1, xshift=0cm] {$\{a_0\}$};
			\node (a1) [right=of a0] {$\{a_1\}$};
			\node (a2) [right=of a1] {$\{a_2\}$};
			\node (dots1) [right=of a2, draw=none] {$\cdots$};
			\node (I2) [right=of I1, xshift=3cm] {$I_2$};
			
			\node (a0a1) [above=of a1] {$\{a0, a1\}$};
			\node (a1a2) [right=of a0a1] {$\{a1, a2\}$};
			\node (dots2) [right=of a1a2, draw=none] {$\cdots$};
			
			\node (dots3) [above=of a0a1, draw=none] {$\cdots$};
			\node (dots4) [above=of a1a2, draw=none] {$\cdots$};
			
			\draw[arrow] (a0a1) -- (a0);
			\draw[arrow] (a1a2) -- (a0a1);
			\draw[arrow] (a2) -- (a1);
			\draw[arrow] (a1) -- (a0);
			\draw[arrow] (a0) -- (I1);
			
			\draw[dashed, thick, arrow] (dots2) -- (a1a2);
			\draw[dashed, thick, arrow] (dots1) -- (a2);
			
			\path [arrow] (I1) edge [out=0, in=45, loop] (I1);
			\path [arrow] (I2) edge [out=0, in=45, loop] (I2);
		\end{tikzpicture}
		\caption{}
	\end{subfigure}
\caption{Stable transition graph (a) and supported transition graph (b) of the NLP \(P\) of Example~\ref{exam:NLP-all}. Herein, \(a_i\) denotes \(p(s^i(0))\), \(I_1(a) = \fval, \forall a \in \hb{P}\) (\(I_1 = \emptyset\)), \(I_2(a) = \tval, \forall a \in \hb{P}\) (\(I_2 = \hb{P}\)).}
\label{fig:exam-NLP-sttg-sptg}
\end{figure}
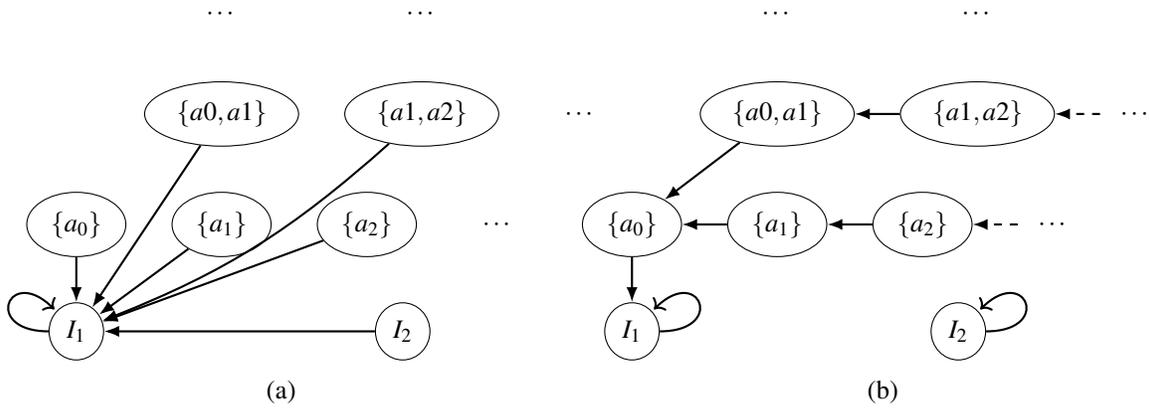

\subsection{Set Theory}\label{subsec:preliminaries-set-theory}

Our proofs in the following sections rely on existing results in set theory.
To keep the paper self-contained, we recall them below.
For more details, we refer the reader to~\cite{Halmos1960}.

\begin{lemma}[Zorn's lemma~\cite{Zorn1935}]\label{lem:Zorn-lemma}
	Given a partially ordered set \((P, \leq)\).
	If any \(\leq\)-chain possesses an upper bound, then \((P, \leq)\) has a maximal element.
\end{lemma}

There is a ``strengthened'' version that is equivalent to Zorn's lemma:

\begin{lemma}[\cite{BS2014}]\label{lem:Zorn-lemma-strengthened}
	Given a partially ordered set \((P, \leq)\), if any \(\leq\)-chain has an upper bound, then for any \(p \in P\), \(P\) has a maximal element \(m\) such that \(p \leq m\).
\end{lemma}

Finally, we show Hartogs' lemma that basically means that for any set \(X\), there is always a bigger ordinal~\cite{Hartogs1915}.

\begin{lemma}[Hartogs' lemma]\label{lem:Hartogs}
	For any set \(X\), there is a least ordinal \(\alpha\) such that there is no injection from \(\alpha\) into \(X\).
\end{lemma}

%% file: sections/ts-def-pro.tex
\section{Trap Space Semantics}\label{sec:TS-defs-properties}

The concept of stable (resp.\ supported) trap set is obtained by relaxing the equality in Definition~\ref{def:NLP-StC-SuC} of stable (resp.\ supported) class, as follows:
\begin{definition}\label{def:NLP-trap-set}
	A non-empty set \(S\) of two-valued interpretations of an \nlp \(P\) is called a \emph{stable trap set} (resp.\ \emph{supported trap set}) of \(P\) \ifftext \(\{\Fop{I} \mid I \in S\} \subseteq S\) (resp.\ \(\{\Top{I} \mid I \in S\} \subseteq S\)).
\end{definition}

Note that a stable (resp.\ supported) class is clearly a stable (resp.\ supported) trap set, but the reverse may not be true.
Indeed, given the \nlp \(P\) of Example~\ref{exam:NLP-all}, \(\{\emptyset\}\) is a stable class and also a stable trap set, whereas \(\{\emptyset, a_0\}\) is a stable trap set but not a stable class.
Similarly, \(\{\emptyset, \hb{P}\}\) is a supported class and also a supported trap set, whereas \(\{\emptyset, a_0, a_1\}\) is a supported trap set but not a supported class.
By adding the condition on three-valued interpretations, we obtain the concept of stable or supported trap space:

\begin{definition}\label{def:NLP-trap-space}
	A three-valued interpretation \(I\) of an \nlp \(P\) is called a \emph{stable trap space} (resp.\ \emph{supported trap space}) of \(P\) \ifftext \(\cset{I}\) is a stable (resp.\ supported) trap set of \(P\).
	We call \(I\) a \emph{trap space} of \(P\) \ifftext it is a stable or supported trap space of \(P\).
\end{definition}

More broadly, it is easy to adapt the concept of stable or supported trap set of an \nlp to a (possibly infinite) directed graph in general:

\begin{definition}\label{def:DiG-trap-set}
	Consider a directed graph \(G\).
	A non-empty subset \(S\) of \(V(G)\) is called a \emph{trap set} of \(G\) \ifftext there are no two vertices \(A\) and \(B\) such that \(A \in S\), \(B \not \in S\), and \((A, B) \in E(G)\).
\end{definition}

By definition, this immediately follows that \(S\) is a stable (resp.\ supported) trap set of \(P\) \ifftext \(S\) is a trap set of \(\tgst{P}\) (resp.\ \(\tgsp{P}\)).
Hence, we can deduce from Definition~\ref{def:NLP-trap-space} that a three-valued interpretation \(I\) is a stable (resp.\ supported) trap space of \(P\) \ifftext \(\cset{I}\) is a trap set of \(\tgst{P}\) (resp.\ \(\tgsp{P}\)).
Since stable and supported transition graphs represent the dynamical aspect of an \nlp~\cite{BS1992,IS2012}, this indicates that the trap space semantics for \nlps admits a dynamical characterization.

\begin{example}\label{exam:NLP-TS}
	Consider the \nlp \(P\) of Example~\ref{exam:NLP-all}.
	The stable and supported transition graphs of \(P\) are shown in Figure~\ref{fig:exam-NLP-sttg-sptg}~(a) and Figure~\ref{fig:exam-NLP-sttg-sptg}~(b), respectively.
	We have that the three-valued interpretation \(I_1\) 
	is a stable (resp.\ supported) trap space of \(P\) because \(\cset{I_1} = \{\emptyset\}\) is a stable (resp.\ supported) trap set of \(P\).
	It is easy to verify that \(\cset{I_1}\) is also a trap set of \(\tgst{P}\) (resp.\ \(\tgsp{P}\)).
	Note that \(\{\emptyset, a_0, a_1\}\) is a stable (resp.\ supported) trap set, but it does not correspond to a stable (resp.\ supported) trap space of \(P\).
	We see that \(I_2\) is a supported trap space but not a stable trap space of \(P\), since \(\cset{I_2} = \{\hb{P}\}\) is a supported trap set but not a stable trap set of \(P\).
\end{example}


We define the \emph{trap space semantics} of an \nlp as the set of its trap spaces.
When applicable we shall restrict ourselves to stable or
supported trap spaces to relate those subsets to other semantics.
For convenience, let \(\stts{P}\) (resp.\ \(\suts{P}\)) denote the set of stable (resp.\ supported) trap spaces of an \nlp \(P\).
We first establish the simplest property, which is, however, vital for proving subsequent results, especially in infinite structures.
Note that, due to space limitations, all proofs are given in Appendix~\ref{sec:appendix-detailed-proofs}.

\begin{propositionE}[][end, restate, text link=]\label{prop:NLP-exist-St-Su-TS}
	An \nlp \(P\) always has a stable or supported trap space.
\end{propositionE}
\begin{proofE}[text proof=]
	Let \(I^*\) be a three-valued interpretation that corresponds to all the two-valued interpretations, i.e., \(\forall a \in \hb{P}, I^*(a) = \star\).
	By setting \(S = \cset{I^*}\), the condition \(\{\Fop{I} \mid I \in S\} \subseteq S\) (resp.\ \(\{\Top{I} \mid I \in S\} \subseteq S\)) always holds.
	Hence, \(\cset{I^*}\) is a stable (resp.\ supported) trap set of \(P\).
	By definition, \(I^*\) is a stable (resp.\ supported) trap space of \(P\).
\end{proofE}

Next, we investigate the consistency among stable or supported trap spaces, which shall be used to investigate the model or class disjointness (see Section~\ref{sec:TS-relationships}).

\begin{definition}\label{def:NLP-interpretation-consistent}
	Consider an \nlp \(P\).
	Two three-valued interpretations \(I_1\) and \(I_2\) of \(P\) are called \emph{consistent} \ifftext for all \(a \in \hb{P}\), \(I_1(a) \leq_s I_2(a)\) or \(I_2(a) \leq_s I_1(a)\).
	Equivalently, \(I_1\) and \(I_2\) are called consistent \ifftext \(\cset{I_1} \cap \cset{I_2} \neq \emptyset\).
\end{definition}

Note that using \(I_1 \leq_s I_2\) or \(I_2 \leq_s I_1\) in Definition~\ref{def:NLP-interpretation-consistent} is insufficient here, since two three-valued interpretations that are not comparable w.r.t. \(\leq_s\) can still be consistent.
We then generalize the above definition to a non-empty set of mutually consistent trap spaces:

\begin{definition}\label{def:NLP-intersection-interpretation-consistent-set}
	Given an \nlp \(P\), let \(E\) be a non-empty set of mutually consistent three-valued interpretations (i.e., \(E_i\) and \(E_j\) are consistent for all \(E_i, E_j \in E\)).
	Then the \emph{intersection} of all elements in \(E\) (denoted by \(\sqcap_E\)) is a three-valued interpretation defined as: \(\sqcap_E(a) = \text{min}_{\leq_s}\{E_i(a) \mid E_i \in E\}\) for all atom \(a \in \hb{P}\).
	This is well-defined because \(\{E_i(a) \mid E_i \in E\}\) is a totally ordered finite set w.r.t. \(\leq_s\) (thus its minimum element exists), since all elements in \(E\) are mutually consistent.
\end{definition}

By pointing out that the intersection of mutually consistent stable (resp.\ supported) trap spaces is a stable (resp.\ supported) trap space, along with applying Definition~\ref{def:NLP-interpretation-consistent} and Lemma~\ref{lem:Zorn-lemma-strengthened}, we obtain the following two lemmas:

\begin{lemmaE}[][end, restate, text link=]\label{lem:NLP-StTS-min-s-StTS}
	Consider an \nlp \(P\) and a three-valued interpretation \(J\) of \(P\).
	If \(J\) is a stable trap space of \(P\), then there exists a \(\leq_s\)-minimal stable trap space \(J'\) of \(P\) such that \(J' \leq_s J\).
\end{lemmaE}
\begin{proofE}[text proof=]
	Clearly, \((\stts{P},  \leq_i)\) is a partially ordered set.
	Consider a \(\leq_i\)-chain \((E_k)_{k \in I}\).
	Note that \(I\) is the set of all elements of the \(\leq_i\)-chain, and it can be written as an indexed family, which can be done via using the set itself as index set.
	Consider now the intersection of all elements in \(I\), said \(\sqcap_I\) whose existence is guaranteed by Definition~\ref{def:NLP-intersection-interpretation-consistent-set}.
	Following from the definition, we have \(\cset{\sqcap_I} = \bigcap_{k \in I}\cset{E_k}\).
	It follows that \(\sqcap_I \leq_s E_k\) for all \(E_k, k \in I\).
	Equivalently, \(E_k \leq_i \sqcap_I\) for all \(E_k, k \in I\), implying that \(\sqcap_I\) is an upper bound w.r.t. \(\leq_i\) of \((E_k)_{k \in I}\).
	It remains to show that \(\sqcap_I\) is a stable trap space of \(P\).
	
	Consider a two-valued interpretation \(A\) of \(P\).
	Let \(B\) be the successor of \(A\) in \(\tgst{P}\), i.e., \(B = \Fop{A}\).
	Assume that \(A \in \cset{\sqcap_I}\).
	Then for all \(k \in I\), we have that \(A \in \cset{E_k}\).
	Thus \(B \in \cset{E_k}\), since \(E_k\) is a stable trap space.
	Hence, \(B \in \cset{\sqcap_I}\).
	This implies that if \(A \in \cset{\sqcap_I}\) then \(\Fop{A} \in \cset{\sqcap_I}\), thus \(\sqcap_I\) is a stable trap set of \(P\).
	Since \(\sqcap_I\) is a three-valued interpretation by definition, it is a stable trap space of \(P\).
	
	By definition, a \(\leq_s\)-minimal stable trap space is equivalent to a \(\leq_i\)-maximal stable trap space.
	By Lemma~\ref{lem:Zorn-lemma-strengthened}, if \(J\) is a stable trap space of \(P\), then there exists a \(\leq_s\)-minimal stable trap space \(J'\) of \(P\) such that \(J \leq_i J'\) (equivalently, \(J' \leq_s J\)).
\end{proofE}

\begin{lemmaE}[][end, restate, text link=]\label{lem:NLP-SuTS-min-s-SuTS}
	Consider an \nlp \(P\) and a three-valued interpretation \(J\) of \(P\).
	If \(J\) is a supported trap space of \(P\), then there exists a \(\leq_s\)-minimal supported trap space \(J'\) of \(P\) such that \(J' \leq_s J\).
\end{lemmaE}
\begin{proofE}[text proof=]
	The proof of Lemma~\ref{lem:NLP-SuTS-min-s-SuTS} can be applied here by substituting \(\Fop{A}\) in Lemma~\ref{lem:NLP-StTS-min-s-StTS} with \(\Top{A}\).
\end{proofE}

In addition to Proposition~\ref{prop:NLP-exist-St-Su-TS}, we obtain the following main result:

\begin{theoremE}[\textbf{main result}][end, restate, text link=]\label{theo:NLP-exist-min-s-StTS-SuTS}
	Any \nlp \(P\) has a \(\leq_s\)-minimal stable (or supported) trap space.
\end{theoremE}
\begin{proofE}[text proof=]
	For the case of \(\leq_s\)-minimal stable trap spaces, this immediately follows from Lemma~\ref{lem:NLP-StTS-min-s-StTS} and the fact that \(P\) has at least one stable trap space by Proposition~\ref{prop:NLP-exist-St-Su-TS}.
	
	For the case of \(\leq_s\)-minimal supported trap spaces, this immediately follows from Lemma~\ref{lem:NLP-SuTS-min-s-SuTS} and the fact that \(P\) has at least one supported trap space by Proposition~\ref{prop:NLP-exist-St-Su-TS}.
\end{proofE}


We then establish the relationship between a non-empty set of two-valued interpretations and stable or supported trap spaces, which is useful for bridging the relationships between stable (resp.\ supported) classes and stable (resp.\ supported) trap spaces later.

\begin{propositionE}[][end, restate, text link=]\label{prop:NLP-unique-minimal-covered-TS}
	Consider an NLP \(P\).
	Let \(S\) be a non-empty set of two-valued interpretations of \(P\).
	Then there is a unique \(\leq_s\)-minimal stable (resp.\ supported) trap space, denoted by \(\langle S\rangle^{st}_{P}\) (resp.\ \(\langle S\rangle^{sp}_{P}\)), covering \(S\), i.e., \(S \subseteq \cset{\langle S\rangle^{st}_{P}}\) (resp.\ \(S \subseteq \cset{\langle S\rangle^{sp}_{P}}\)).
\end{propositionE}
\begin{proofE}[text proof=]
	Hereafter, we prove the case of stable trap spaces. 
	The proof for the case of supported trap spaces is analogous.
	Let \(T\) be the set of all stable trap spaces \(I\) such that \(S \subseteq \cset{I}\).
	The set \(T\) is non-empty as it contains \(\epsilon\) that is a stable trap space in which all atoms are assigned to \(\uval\) (see Proposition~\ref{prop:NLP-exist-St-Su-TS}).
	Since \(S\) is non-empty, the elements in \(T\) are mutually consistent, thus we can take the intersection of all these elements said \(\sqcap_T\).
	
	Next, we show that \(\sqcap_T\) is also a stable trap space of \(P\).
	Consider a two-valued interpretation \(A\) of \(P\).
	Let \(B\) be the successor of \(A\) in \(\tgst{P}\), i.e., \(B = \Fop{A}\).
	Assume that \(A \in \cset{\sqcap_T}\).
	Then for all \(E_k \in T\), we have that \(A \in \cset{E_k}\).
	Thus \(B \in \cset{E_k}\), since \(E_k\) is a stable trap space.
	Hence, \(B \in \cset{\sqcap_T}\).
	This implies that if \(A \in \cset{\sqcap_T}\) then \(\Fop{A} \in \cset{\sqcap_T}\), thus \(\sqcap_T\) is a stable trap set of \(P\).
	Since \(\sqcap_T\) is a three-valued interpretation by definition, it is a stable trap space of \(P\).
	
	By definition, \(\sqcap_T\) is unique and \(\leq_s\)-minimal.
	Since \(S \subseteq \cset{E_k}\) for all \(E_k \in T\), \(S \subseteq \cset{\sqcap_T}\).
	By setting \(\langle S\rangle^{st}_{P} = \sqcap_T\), we can conclude the proof.
\end{proofE}

Finally, we show the separation of trap spaces \wrttext the \(\leq_s\)-minimality:

\begin{propositionE}[][end, restate, text link=]\label{prop:NLP-two-min-ts-separation}
	Consider an \nlp \(P\).
	Let \(I_1\) and \(I_2\) be two distinct \(\leq_s\)-minimal stable (resp.\ supported) trap spaces of \(P\).
	Then \(I_1\) and \(I_2\) are not consistent.
\end{propositionE}
\begin{proofE}[text proof=]
	Hereafter, we prove the case of stable trap spaces. 
	The proof for the case of supported trap spaces is analogous.
	Assume that \(I_1\) and \(I_2\) are consistent.
	Then \(I = \sqcap_{\{I_1, I_2\}}\) exists and \(\cset{I} = \cset{I_1} \cap \cset{I_2}\).
	By similarly applying the arguments of Proposition~\ref{prop:NLP-unique-minimal-covered-TS}, we have \(I\) is also a stable trap space of \(P\).
	Since \(I_1\) and \(I_2\) are distinct, we have \(\cset{I_1} \subset \cset{I}\) or \(\cset{I_2} \subset \cset{I}\) must hold.
	This leads to a contradiction because \(I_1\) and \(I_2\) are \(\leq_s\)-minimal stable trap spaces of \(P\).
	Hence, \(I_1\) and \(I_2\) are not consistent.
\end{proofE}

%% file: sections/relationship.tex
\section{Relationships with Other Semantics}\label{sec:TS-relationships}

Naturally, if a stable (resp.\ supported) trap space of an \nlp \(P\) is two-valued, then it is also a stable (resp.\ supported) model of \(P\).
We now present additional relationships between stable or supported trap spaces and other types of models.
These results show that the trap space semantics can provide a \emph{unified} semantic framework for \nlps.
Note that there are historically different model-theoretic or dynamical semantics for \nlps, and they are often defined in different ways~\cite{YY1995,SZ1997}.
For convenience, let \(\stpms{P}\) (resp.\ \(\supms{P}\)) denote the set of stable (resp.\ supported) partial models of \(P\); let \(\stms{P}\) (resp.\ \(\sums{P}\)) denote the set of stable (resp.\ supported) models of \(P\); and let \(\regms{P}\) denote the set of regular models of \(P\).

\subsection{Relationships with the Stable and Supported Class Semantics}\label{subsec:TS-relationships-StC-SuC}

By leveraging the concept of stable or supported trap set underlying the trap space semantics, we correct a flaw in the existing alternative characterization of strict stable or supported classes proved in~\cite{IS2012}.
Due to space limitations, we refer the reader to Section~\ref{sec:char-st-su-class} for the technical details.

We now prove a similar version of Lemma~\ref{lem:NLP-StTS-min-s-StTS}, but for stable trap sets instead of stable trap spaces.
The proof technique is analogous to that of Lemma~\ref{lem:NLP-StTS-min-s-StTS}.

\begin{lemmaE}[][end, restate, text link=]\label{lem:NLP-St-trap-set-min-s}
	Consider an \nlp \(P\) and a non-empty set \(S\) of two-valued interpretations of \(P\).
	If \(S\) is a stable trap set of \(P\), then there exists a \(\subseteq\)-minimal stable trap set \(S'\) of \(P\) such that \(S' \subseteq S\).
\end{lemmaE}
\begin{proofE}[text proof=]
	Let \(TS(P)\) denote the set of all stable trap sets of \(P\).
	Clearly, \((TS(P),  \supseteq)\) is a partially ordered set.
	Given a \(\supseteq\)-chain \((E_i)_{i \in I}\).
	Note that \(I\) is the set of all elements of the \(\supseteq\)-chain, and it can be written as an indexed family, which can be done via using the set itself as index set.
	Consider now \(E = \bigcap_{i \in I}E_i\).
	Obviously, \(E\) is an upper bound of \((E_i)_{i \in I}\), i.e., \(E_i \supseteq E\), for all \(i \in I\).
	It remains to show that \(E\) is a stable trap set of \(P\).
	
	Consider a two-valued interpretation \(A\).
	\(A \in E\) \ifftext \(A \in E_i, \forall i \in I\).
	This implies that \(\Fop{A} \in E_i, \forall i \in I\) since \(E_i\) is a stable trap set of \(P\).
	Then \(\Fop{A} \in E\).
	This implies that \(\{\Fop{B} \mid B \in E\} \subseteq E\).
	By definition, \(E\) is a stable trap set of \(P\).
	
	By definition, \(\subseteq\)-minimal stable trap set is equivalent to a \(\supseteq\)-maximal stable trap set.
	By Lemma~\ref{lem:Zorn-lemma-strengthened}, if \(S\) is a stable trap set of \(P\), then there exists a \(\subseteq\)-minimal stable trap set \(S'\) of \(P\) such that \(S \supseteq S'\) (equivalently, \(S' \subseteq S\)).
\end{proofE}

This is similar for supported trap sets:

\begin{lemmaE}[][end, restate, text link=]\label{lem:NLP-Su-trap-set-min-s}
	Given an \nlp \(P\), if a non-empty set \(S\) of two-valued interpretations of \(P\) is a supported trap set of \(P\), then there exists a \(\subseteq\)-minimal supported trap set \(S'\) of \(P\) such that \(S' \subseteq S\).
\end{lemmaE}
\begin{proofE}[text proof=]
	The proof of Lemma~\ref{lem:NLP-St-trap-set-min-s} can be applied here by substituting \(\Fop{I}\) in Lemma~\ref{lem:NLP-St-trap-set-min-s} with \(\Top{I}\).
\end{proofE}

Based on Lemma~\ref{lem:NLP-St-trap-set-min-s} 
, we obtain the following main result:

\begin{theoremE}[\textbf{main result}][end, restate, text link=]\label{theo:NLP-min-St-trap-set-strict-StC}
	Consider an \nlp \(P\).
	A non-empty set \(S\) of two-valued interpretations is a \(\subseteq\)-minimal stable trap set of \(P\) \ifftext \(S\) is a strict stable class of \(P\).
\end{theoremE}
\begin{proofE}[text proof=]
	``\(\Rightarrow\)'' Assume that \(S\) is a \(\subseteq\)-minimal stable trap set of \(P\).
	Then \(\{\Fop{I} \mid I \in S\} \subseteq S\) by definition.
	Suppose that \(\{\Fop{I} \mid I \in S\} \subsetneq S\).
	Then there exists \(J \in S\) such that for all \(I \in S\), \(\Fop{I} \neq J\).
	We have \(\{\Fop{I} \mid I \in S \setminus \{J\}\} \subseteq \{\Fop{I} \mid I \in S\} \subseteq S \setminus \{J\}\).
	Thus \(S \setminus \{J\}\) is a stable trap set of \(P\), which contradicts the \(\subseteq\)-minimality of \(S\).
	Hence \(\{\Fop{I} \mid I \in S\} = S\), i.e., \(S\) is a stable class of \(P\).
	Since a stable class is a stable trap set, \(S\) must be a strict stable class.
	
	``\(\Leftarrow\)'' Assume that \(S\) is a strict stable class of \(P\).
	Then \(\{\Fop{I} \mid I \in S\} = S\) by definition.
	This implies that \(S\) is a stable trap set of \(P\).
	By Lemma~\ref{lem:NLP-St-trap-set-min-s}, there is a \(\subseteq\)-minimal stable trap set \(S'\) such that \(S' \subseteq S\).
	We have \(\{\Fop{I} \mid I \in S'\} \subseteq S'\) by definition.
	Suppose that \(\{\Fop{I} \mid I \in S'\} \subsetneq S'\).
	Then there is \(J \in S'\) such that \(\Fop{I} \neq J\) for all \(I \in S'\).
	We have \(\{\Fop{I} \mid I \in S' \setminus \{J\}\} \subseteq \{\Fop{I} \mid I \in S'\} \subseteq S' \setminus \{J\}\).
	By definition, \(S' \setminus \{J\}\) is a stable trap set of \(P\), which contradicts the \(\subseteq\)-minimality of \(S'\).
	Thus, \(\{\Fop{I} \mid I \in S'\} = S'\), i.e., \(S'\) is a stable class of \(P\).
	\(S\) is a strict stable class of \(P\), thus \(S = S'\) because of the \(\subseteq\)-minimality of a strict stable class.
	Hence, \(S\) is a \(\subseteq\)-minimal stable trap set of \(P\).
\end{proofE}

Analogously, based on Lemma~\ref{lem:NLP-Su-trap-set-min-s} 
, we obtain the following:

\begin{theoremE}[][end, restate, text link=]\label{theo:NLP-min-Su-trap-set-strict-SuC}
	Consider an \nlp \(P\).
	A non-empty set \(S\) of two-valued interpretations is a \(\subseteq\)-minimal supported trap set of \(P\) \ifftext \(S\) is a strict supported class of \(P\).
\end{theoremE}
\begin{proofE}[text proof=]
	The proof of Theorem~\ref{theo:NLP-min-St-trap-set-strict-StC} can be applied here by substituting \(\Fop{I}\) and Lemma~\ref{lem:NLP-St-trap-set-min-s} in Theorem~\ref{theo:NLP-min-St-trap-set-strict-StC} with \(\Top{I}\) and Lemma~\ref{lem:NLP-Su-trap-set-min-s}, respectively.
\end{proofE}

Theorem~\ref{theo:NLP-min-St-trap-set-strict-StC} and Theorem~\ref{theo:NLP-min-Su-trap-set-strict-SuC} immediately entail that every stable (resp.\ supported) trap space always covers at least one strict stable (resp.\ supported) class in an \nlp:

\begin{corollaryE}[][end, restate, text link=]\label{cor:NLP-min-TS-contain-strict-class}
	Consider an \nlp \(P\).
	For every stable (resp.\ supported) trap space \(I\) of \(P\), \(\cset{I}\) contains at least one strict stable (resp.\ supported) class of \(P\).
\end{corollaryE}
\begin{proofE}[text proof=]
	Assume that \(I\) is a stable trap space of \(P\).
	By definition, \(\cset{I}\) is a stable trap set of \(P\).
	There exists a \(\subseteq\)-minimal stable trap set \(S' \subseteq \cset{I}\).
	By Theorem~\ref{theo:NLP-min-St-trap-set-strict-StC}, \(S'\) is also a strict stable class of \(P\).
	
	Assume that \(I\) is a supported trap space of \(P\).
	By definition, \(\cset{I}\) is a supported trap set of \(P\).
	There exists a \(\subseteq\)-minimal supported trap set \(S' \subseteq \cset{I}\).
	By Theorem~\ref{theo:NLP-min-Su-trap-set-strict-SuC}, \(S'\) is also a strict supported class of \(P\).
\end{proofE}

They also naturally entail the disjointness of two strict stable or supported classes:

\begin{propositionE}[][end, restate, text link=]\label{prop:NLP-disjoint-strict-classes}
	Two distinct strict stable (resp.\ supported) classes of an \nlp \(P\) are disjoint.
\end{propositionE}
\begin{proofE}[text proof=]
	We prove the case of strict stable classes first.
	
	Let \(S_1\) and \(S_2\) be two distinct strict stable classes of \(P\).
	By Theorem~\ref{theo:NLP-min-St-trap-set-strict-StC}, \(S_1\) and \(S_2\) are two distinct \(\subseteq\)-minimal stable trap sets of \(P\).
	Assume that \(S_1\) and \(S_2\) are not disjoint.
	It follows that \(S_1 \cap S_2 \neq \emptyset\).
	Consider \(I \in S_1 \cap S_2\).
	We have \(I \in S_1\) and \(I \in S_2\), then \(\Fop{I} \in S_1\) and \(\Fop{I} \in S_2\) because \(S_1\) and \(S_2\) are stable trap sets, thus \(\Fop{I} \in S_1 \cap S_2\).
	It follows that \(S_1 \cap S_2\) is a stable trap set of \(P\).
	Then \(S_1 \cap S_2 \subsetneq S_1\) or \(S_1 \cap S_2 \subsetneq S_2\) must hold because \(S_1 \neq S_2\), which contradicts the \(\subseteq\)-minimality of \(S_1\) and \(S_2\).
	Hence, \(S_1\) and \(S_2\) are disjoint.
	
	For the case of strict supported classes, we only need to replace \(F_P\) and Theorem~\ref{theo:NLP-min-St-trap-set-strict-StC} with \(T_P\) and Theorem~\ref{theo:NLP-min-Su-trap-set-strict-SuC}, respectively.
\end{proofE}

In~\cite{BS1992}, the authors prove the disjointness of two strict stable classes in \pnames only.
In~\cite{IS2012}, the authors prove the disjointness of two strict supported classes in \pnames only.

\subsection{Relationships with the Stable and Supported Partial Model Semantics}\label{subsec:TS-relationships-StPM-SuPM}

We first establish a model-theoretic characterization for supported trap spaces in an \nlp.

\begin{theoremE}[\textbf{main result}][end, restate, text link=]\label{theo:NLP-char-SuTS}
	Given an \nlp \(P\), \(I\) is a supported trap space of \(P\) \ifftext \(I(\rhs{a}) \leq_s I(a)\) for every \(a \in \hb{P}\).
\end{theoremE}
\begin{proofE}[text proof=]
	``\(\Rightarrow\)'' Assume that \(I\) is a supported trap space of \(P\).
	Suppose that there is an atom \(a \in \hb{P}\) such that \(I(a) <_s I(\rhs{a})\).
	By the definition of the order \(\leq_s\), we have \(I(\rhs{a}) = \uval\) and \(I(a) \neq \uval\).
	Since \(I(\rhs{a}) = \uval\), there exists a two-valued interpretation \(J \in \cset{I}\) such that \(J(\rhs{a}) = \neg I(a)\).
	By definition, \(\Top{J}(a) = J(\rhs{a})\), thus \(\Top{J}\) does not belong to \(\cset{I}\).
	This implies that \(\cset{I}\) is not a supported trap set by definition, which is a contradiction.
	Hence, \(I(\rhs{a}) \leq_s I(a)\) for every \(a \in \hb{P}\).
	
	``\(\Leftarrow\)'' Assume that \(I(\rhs{a}) \leq_s I(a)\) for every atom \(a \in \hb{P}\).
	Let \(A\) be a two-valued interpretation in \(\cset{I}\).
	Of course, \(A \leq_s I\).
	By the characteristics of the three-valued logic, we have \(A(\rhs{a}) \leq_s I(\rhs{a})\), then \(A(\rhs{a}) \leq_s I(a)\) for every \(a \in \hb{P}\).
	This implies that \(\Top{A} \leq_s I\), thus \(\Top{A} \in \cset{I}\).
	Hence, \(\cset{I}\) is a supported trap set of \(P\).
	By definition, \(I\) is a supported trap space of \(P\).
\end{proofE}

Theorem~\ref{theo:NLP-char-SuTS} also points out that a supported trap space may not be a three-valued model of \(\comp{P}\), as it considers the order \(\leq_s\), whereas the completion considers the equivalence \(\leftrightarrow\).
Indeed, given the \nlp \(P\) of Example~\ref{exam:NLP-all}, let \(I_4\) be the three-valued interpretation where \(I_4(a_0) = \uval\) and \(I_4(a) = \fval\) for every \(a \in \hb{P} \setminus \{a_0\}\).
It is easy to verify that \(I_4\) is a supported trap space but not a three-valued model of \(\comp{P}\).

To reason about stable trap spaces, we apply the least fixpoint transformation, which always produces a negative \nlp.
We then show that for negative \nlps, the stable transition graph coincides with the supported transition graph, entailing the coincidence between stable and supported trap spaces.

\begin{theoremE}[][end, restate, text link=]\label{theo:neg-NLP-StTG-SuTG}
	The stable and supported transition graphs of a negative \nlp \(P\) are the same.
\end{theoremE}
\begin{proofE}[text proof=]
	It has been shown in~\cite{IS2012} that for any two-valued interpretation \(I\), \(\Fop{I} = \Top{I}\), if \(P\) is a negative \nlp.
	By the definitions of stable and supported transition graphs, we have that \(\tgst{P} = \tgsp{P}\).
\end{proofE}

\begin{corollaryE}[][end, restate, text link=]\label{cor:neg-NLP-StTS-SuTS}
	Let \(P\) be a negative \nlp.
	Then the set of stable trap spaces of \(P\) coincides with the set of supported trap spaces of \(P\).
\end{corollaryE}
\begin{proofE}[text proof=]
	This immediately follows from Theorem~\ref{theo:neg-NLP-StTG-SuTG}.
\end{proofE}

\begin{theorem}[\cite{YY1994}]\label{theo:neg-NLP-StPM-SuPM}
	Given a negative \nlp \(P\), then the set of stable partial models of \(P\) coincides with the set of supported partial models of \(P\).
\end{theorem}

\begin{theorem}[Theorem 3.1 of~\cite{AD1995}]\label{theo:NLP-lfp-model-equivalence}
	Consider an \nlp \(P\).
	Let \(\lfp{P}\) be the least fixpoint of \(P\).
	Then \(P\) and \(\lfp{P}\) have the same set of stable partial models, the same sets of regular models, and the same set of stable models.
\end{theorem}

Based on the property that \(\Fop{I} = F_{\lfp{P}}(I)\) for any two-valued interpretation \(I\), proved in~\cite{IS2012}, we directly obtain the following:

\begin{theoremE}[][end, restate, text link=]\label{theo:NLP-lfp-StTG}
	Let \(P\) be an \nlp and \(\lfp{P}\) denote the least fixpoint of \(P\).
	Then \(\tgst{P} = \tgst{\lfp{P}}\), i.e., \(P\) and \(\lfp{P}\) have the same stable transition graph.
\end{theoremE}
\begin{proofE}[text proof=]
	It has been shown in~\cite{IS2012} that for any two-valued interpretation \(I\), we have that \(\Fop{I} = F_{\lfp{P}}(I)\).
	By the definition of stable transition graph, we have that \(\tgst{P} = \tgst{\lfp{P}}\).
\end{proofE}

Since stable trap spaces are defined based on the stable transition graph, the above result immediately implies the following: 

\begin{corollaryE}[][end, restate, text link=]\label{cor:NLP-lfp-StTS}
	Let \(P\) be an \nlp and \(\lfp{P}\) denote the least fixpoint of \(P\).
	Then \(P\) and \(\lfp{P}\) have the same set of stable trap spaces.
\end{corollaryE}
\begin{proofE}[text proof=]
	This immediately follows from Theorem~\ref{theo:NLP-lfp-StTG}.
\end{proofE}

By Theorem~\ref{theo:NLP-char-SuTS}, it is easy to derive the following simple relationship between supported partial models and supported trap spaces of an \nlp:

\begin{propositionE}[][end, restate, text link=]\label{prop:NLP-SuPM-is-SuTS}
	Given an \nlp \(P\), if \(I\) is a supported partial model of \(P\), then it is also a supported trap space of \(P\).
\end{propositionE}
\begin{proofE}[text proof=]
	Assume that \(I\) is a supported partial model of \(P\).
	Then \(I\) is a three-valued model of \(\comp{P}\).
	It follows that \(I(a) = I(\rhs{a})\), \(\forall a \in \hb{P}\).
	Thus, \(I\) also satisfies the condition \(I(\rhs{a}) \leq_s I(a)\), \(\forall a \in \hb{P}\).
	By Theorem~\ref{theo:NLP-char-SuTS}, \(I\) is also a supported trap space of \(P\).
\end{proofE}

By applying the least fixpoint transformation, we obtain a similar relationship:

\begin{propositionE}[][end, restate, text link=]\label{prop:NLP-StPM-is-StTS}
	Given an \nlp \(P\), if \(I\) is a stable partial model of \(P\), then it is also a stable trap space of \(P\).
\end{propositionE}
\begin{proofE}[text proof=]
	Assume that \(I\) is a stable partial model of \(P\). \\
	Then \(I\) is also a stable partial model of \(\lfp{P}\) by Theorem~\ref{theo:NLP-lfp-model-equivalence}. \\
	Since \(\lfp{P}\) is a negative \nlp, \(I\) is also a supported partial model of \(\lfp{P}\) by Theorem~\ref{theo:neg-NLP-StPM-SuPM}. \\
	By Proposition~\ref{prop:NLP-SuPM-is-SuTS}, \(I\) is a supported trap space of \(\lfp{P}\). \\
	By Corollary~\ref{cor:neg-NLP-StTS-SuTS}, \(I\) is also a stable trap space of \(\lfp{P}\). \\
	By Corollary~\ref{cor:NLP-lfp-StTS}, \(I\) is a stable trap space of \(P\).
\end{proofE}

Next, we establish a deeper relationship between supported trap spaces and supported partial models.
We start with defining a new concept as follows:

\begin{definition}\label{def:NLP-next-SuTS}
	Consider an \nlp \(P\).
	We define a mapping \(f_P \colon \suts{P} \to \suts{P}\) as follows: if \(m\) is a supported trap space, then \(f_P(m)(v) = m(\rhs{v}), \forall v \in \hb{P}\).
	The mapping \(f_P\) is well-defined because \(f_P(m)\) is always a supported trap space if \(m\) is a supported trap space (see Lemma~\ref{lem:NLP-next-SuTS}).
\end{definition}

We then show that the mapping \(f_{P}\) preserves the supported trap space property:

\begin{lemmaE}[][end, restate, text link=]\label{lem:NLP-next-SuTS}
	Given an \nlp \(P\) and a supported trap space \(I\) of \(P\),  \(f_P(I)\) is a supported trap space of \(P\) and \(I \leq_i f_P(I)\).
\end{lemmaE}
\begin{proofE}[text proof=]
	By definition, \(f_P(I)(v) = I(\rhs{v})\) for all \(v \in \hb{P}\).
	Since \(I\) is a supported trap space of \(P\), \(I(v) \leq_i I(\rhs{v})\) for all \(v \in \hb{P}\) by Theorem~\ref{theo:NLP-char-SuTS}.
	Thus, \(I \leq_i f_P(I)\).
	Then, by the truth table of the three-valued logic, \(I(\rhs{v}) \leq_i f_P(I)(\rhs{v})\) for all \(v \in \hb{P}\).
	This implies that \(f_P(I)(v) \leq_i f_P(I)(\rhs{v})\) for all \(v \in \hb{P}\).
	Hence, \(f_P(I)\) is a supported trap space of \(P\) by Theorem~\ref{theo:NLP-char-SuTS}.
\end{proofE}

By combining Lemma~\ref{lem:NLP-next-SuTS} and the use of Hartogs' lemma (see Lemma~\ref{lem:Hartogs}), we obtain the following main result:

\begin{theoremE}[\textbf{main result}][end, restate, text link=]\label{theo:NLP-SuTS-reach-unique-SuPM}
	Consider an \nlp \(P\).
	For every supported trap space \(I \in \suts{P}\), there is a unique supported partial model \(I' \in \supms{P}\) such that \(I' \leq_s I\).
\end{theoremE}
\begin{proofE}[text proof=]
	We define a function \(K\) recursively on the ordinals as follows:
	\begin{align*}
		K(0) &= I \\
		K(\alpha + 1) &= f_P(K(\alpha)).
	\end{align*}
	If \(\beta\) is a limit ordinal, then by construction and Lemma~\ref{lem:NLP-next-SuTS} \(\{K(\alpha) \colon \alpha < \beta\}\) is a chain w.r.t. \(\leq_i\) in \(\suts{P}\).
	Define \(K(\beta) = \text{sup}_{\leq_i}\{K(\alpha) \colon \alpha < \beta\}\), where \(\text{sup}_{\leq_i}\) denotes the least upper bound w.r.t. \(\leq_i\).
	We have \(K\) is an increasing function from the ordinals into \(\suts{P}\).
	It cannot be strictly increasing, as if it was, then we would have an injection from the ordinals into \(\suts{P}\), which violates Lemma~\ref{lem:Hartogs}.
	Hence, \(K\) must be eventually constant.
	It follows that for some ordinal \(\alpha\), \(K(\alpha + 1) = K(\alpha)\), then \(f_P(K(\alpha)) = K(\alpha)\).
	Let \(\alpha^*\) be the least ordinal such that \(K(\alpha^* + 1) = K(\alpha^*)\).
	We have \(K(\beta) = K(\alpha^*)\) for all \(\beta > \alpha^*\).
	Hence, \(K(\alpha^*)\) is unique.
	Set \(I' = K(\alpha^*)\).
	We have \(f_P(I') = I'\), i.e., \(I'\) is a three-valued model of \(\comp{P}\).
	Therefore, \(I' \in \supms{P}\).
	By construction and Lemma~\ref{lem:NLP-next-SuTS}, \(I \leq_i I'\).
	This implies that \(I' \leq_s I\).
\end{proofE}

The beauty of Theorem~\ref{theo:NLP-SuTS-reach-unique-SuPM} is that the function \(f_{P}\) delivers a unique supported partial model starting from any supported trap space.
This insight is key to prove the following results.

\begin{corollaryE}[][end, restate, text link=]\label{cor:NLP-min-s-SuPM-min-s-SuTS}
	Consider an \nlp \(P\).
	A three-valued interpretation \(I\) is a \(\leq_s\)-minimal supported partial model of \(P\) \ifftext \(I\) is a \(\leq_s\)-minimal supported trap space of \(P\).
\end{corollaryE}
\begin{proofE}[text proof=]
	Recall that \(\suts{P}\) (resp.\ \(\supms{P}\)) denotes the set of all supported trap spaces (resp.\ partial models) of \(P\).
	We have \(\supms{P} \subseteq \suts{P}\) by Proposition~\ref{prop:NLP-SuPM-is-SuTS}.
	
	``\(\Rightarrow\)'' Assume that \(I\) is not a \(\leq_s\)-minimal supported trap space of \(P\).
	Then there exists \(I' \in \suts{P}\) such that \(I' <_s I\).
	By Theorem~\ref{theo:NLP-SuTS-reach-unique-SuPM}, there exists \(J \in \supms{P}\) such that \(J \leq_s I'\).
	This implies that \(J <_s I\), which contradicts the \(\leq_s\)-minimality of \(I\).
	Hence, \(I\) is a \(\leq_s\)-minimal supported trap space of \(P\).
	
	``\(\Leftarrow\)'' This direction is trivial because \(\supms{P} \subseteq \suts{P}\).
\end{proofE}

Corollary~\ref{cor:NLP-min-s-SuPM-min-s-SuTS} is insightful as it indicates that, although the set of supported partial models may be different from the set of supported trap spaces, their \(\leq_s\)-minimal elements are the same.
Hereafter, we show an interesting implication of the above relationships between supported trap spaces and supported partial models.
We start by proving that a supported class that is included in a supported trap space is also included in the application of \(f_{P}\) to this supported trap space:

\begin{lemmaE}[][end, restate, text link=]\label{lem:NLP-SuC-in-next-interpretation}
	Consider an \nlp \(P\).
	Let \(I\) be a supported trap space of \(P\).
	If \(S\) is a supported class of \(P\) such that \(S \subseteq \cset{I}\), then \(S \subseteq \cset{f_{P}(I)}\).
\end{lemmaE}
\begin{proofE}[text proof=]
	Consider \(J \in S\).
	Since \(S \subseteq \cset{I}\), we have \(J \in \cset{I}\).
	Let \(J' = \Top{J}\).
	For any atom \(a \in \hb{P}\), we have \(J'(a) = J(\rhs{a}) \leq_s I(\rhs{a}) = f_{P}(I)(a)\).
	This implies that \(J' \in \cset{f_{P}(I)}\).
	Hence, \(\{\Top{J} \mid J \in S\} \subseteq \cset{f_{P}(I)}\).
	Since \(S\) is a supported class, \(\{\Top{J} \mid J \in S\} = S\) by definition.
	This implies that \(S \subseteq \cset{f_{P}(I)}\).
\end{proofE}

We now prove the case of supported classes:

\begin{propositionE}[][end, restate, text link=]\label{prop:NLP-SuPM-cover-SuC}
	Given an NLP \(P\), if \(S\) is a supported class, then \(\langle S\rangle^{sp}_{P}\) is a supported partial model.
\end{propositionE}
\begin{proofE}[text proof=]
	By Proposition~\ref{prop:NLP-unique-minimal-covered-TS}, \(\langle S\rangle^{sp}_{P}\) is a supported trap space of \(P\).
	Assume that it is not a supported partial model of \(P\).
	Then by definition, we have \(f_{P}(\langle S\rangle^{sp}_{P}) \neq \langle S\rangle^{sp}_{P}\).
	By Lemma~\ref{lem:NLP-next-SuTS}, \(f_{P}(\langle S\rangle^{sp}_{P}) \leq_s \langle S\rangle^{sp}_{P}\).
	It follows that \(f_{P}(\langle S\rangle^{sp}_{P}) <_s \langle S\rangle^{sp}_{P}\).
	By Lemma~\ref{lem:NLP-SuC-in-next-interpretation}, \(S \subseteq \cset{f_{P}(\langle S\rangle^{sp}_{P})}\) because \(S\) is a supported class and \(S \subseteq \cset{\langle S\rangle^{sp}_{P}}\) by definition.
	This implies that \(\langle S\rangle^{sp}_{P}\) is not a \(\leq_s\)-minimal supported trap space covering \(S\), which is a contradiction.
	Hence, \(\langle S\rangle^{sp}_{P}\) is a supported partial model of \(P\).
\end{proofE}

This is similar for stable classes:

\begin{propositionE}[][end, restate, text link=]\label{prop:NLP-StPM-cover-StC}
	Given an NLP \(P\), if \(S\) is a stable class of \(P\), then \(\langle S\rangle^{st}_{P}\) is a stable partial model.
\end{propositionE}
\begin{proofE}[text proof=]
	By Theorem~\ref{theo:NLP-lfp-StTG} and Theorem~\ref{theo:neg-NLP-StTG-SuTG}, 
	\(\tgst{P} = \tgst{\lfp{P}} = \tgsp{\lfp{P}}\).
	Hence, \(S\) is a supported class of \(\lfp{P}\).
	In addition, \(\langle S\rangle^{st}_{P} = \langle S\rangle^{sp}_{\lfp{P}}\).
	By Proposition~\ref{prop:NLP-SuPM-cover-SuC}, \(\langle S\rangle^{sp}_{\lfp{P}}\) is a supported partial model of \(\lfp{P}\).
	It is also a stable partial model of \(\lfp{P}\) by Theorem~\ref{theo:neg-NLP-StPM-SuPM}, since \(\lfp{P}\) is negative.
	By Theorem~\ref{theo:NLP-lfp-model-equivalence}, \(\langle S\rangle^{st}_{P}\) is a stable partial model of \(P\).
\end{proofE}

Proposition~\ref{prop:NLP-SuPM-cover-SuC} (resp.\ Proposition~\ref{prop:NLP-StPM-cover-StC}) gives us deeper understanding about relationships between the supported (resp.\ stable) classes and supported (resp.\ stable) partial models of an \nlp.

\subsection{Relationships with the Regular Model Semantics}\label{subsec:TS-relationships-RegM}

Recall that a stable trap space may not be a stable partial model.
Indeed, given the \nlp \(P\) of Example~\ref{exam:NLP-all}, let \(I_4\) be the three-valued interpretation where \(I_4(a_0) = \uval\) and \(I_4(a) = \fval\) for every \(a \in \hb{P} \setminus \{a_0\}\).
It is easy to verify that \(I_4\) is a stable trap space of \(P\), but it is not a stable partial model of \(P\).
However, since every two-valued interpretation of \(P\) has an arc to \(I_1\) 
in the stable transition graph of \(P\), for every stable trap space \(I\) of \(P\), \(I(a) \neq 1\) for every \(a \in \hb{P}\).
This implies that \(I_1\) is the \(\leq_s\)-minimal stable trap space of \(P\).
Since \(I_1\) is also the regular model of \(P\), we wonder if the set of regular models of \(P\) coincides with the set of \(\leq_s\)-minimal stable trap spaces of \(P\), which if true is really interesting.
We answer this question for \nlps as follows:

\begin{theoremE}[\textbf{main result}][end, restate, text link=]\label{theo:NLP-RegM-s-min-StTS}
	Given an \nlp \(P\), a three-valued interpretation \(I\) is a regular model of \(P\) \ifftext \(I\) is a \(\leq_s\)-minimal stable trap space of \(P\).
\end{theoremE}
\begin{proofE}[text proof=]
	The programs \(P\) and \(\lfp{P}\) have the same set of regular models by Theorem~\ref{theo:NLP-lfp-model-equivalence}.
	They also have the same set of stable trap spaces by Corollary~\ref{cor:NLP-lfp-StTS}.
	Hence, it suffices to show that \(I\) is a regular model of \(\lfp{P}\) \ifftext \(I\) is a \(\leq_s\)-minimal stable trap space of \(\lfp{P}\).
	
	Recall that \(\lfp{P}\) is a negative NLP.
	Then \(I\) is a regular model of \(\lfp{P}\) \\
	\ifftext \(I\) is a \(\leq_s\)-minimal stable partial model of \(\lfp{P}\) by definition \\
	\ifftext \(I\) is a \(\leq_s\)-minimal supported partial model of \(\lfp{P}\) by Theorem~\ref{theo:neg-NLP-StPM-SuPM} \\
	\ifftext \(I\) is a \(\leq_s\)-minimal supported trap space of \(\lfp{P}\) by Corollary~\ref{cor:NLP-min-s-SuPM-min-s-SuTS} \\
	\ifftext \(I\) is a \(\leq_s\)-minimal stable trap space of \(\lfp{P}\) by Corollary~\ref{cor:neg-NLP-StTS-SuTS}.
\end{proofE}

By Corollary~\ref{cor:NLP-min-TS-contain-strict-class}, Proposition~\ref{prop:NLP-two-min-ts-separation}, and Theorem~\ref{theo:NLP-RegM-s-min-StTS}, each regular model of an \nlp contains at least one strict stable class of this \nlp, which is a new insight to the best of our knowledge.

\subsection{Relationships with the L-stable Model Semantics}\label{subsec:TS-relationships-L-StM}

We first define a new order on \(\threed{}\), which considers only the undefined value \(\uval\):

\begin{definition}\label{def:NLP-u-order}
	The order \(\leq_u\) on \(\threed{}\) is given by \(0 \leq_u 1\), \(1 \leq_u 0\), \(0 <_u \star\), \(1 <_u \star\), and there is no other relation.
	Given two three-valued interpretations \(I_1\) and \(I_2\) of an \nlp \(P\), \(I_1 \leq_u I_2\) \ifftext \(I_1(v) \leq_u I_2(v)\) for all \(v \in \hb{P}\).
	It is easy to see that \(I_1 \leq_s I_2\) implies \(I_1 \leq_u I_2\).
\end{definition}

\begin{propositionE}[][end, restate, text link=]\label{prop:NLP-u-order-subset}
	Given two three-valued interpretations \(I_1\) and \(I_2\) of an \nlp \(P\), \(I_1 \leq_u I_2\) \ifftext \(I_1^{\textbf{u}} \subseteq I_2^{\textbf{u}}\).
\end{propositionE}
\begin{proofE}[text proof=]
	We have \(I_1 \leq_u I_2\) \\
	\ifftext \(I_1(v) \leq_u I_2(v)\) for all \(v \in \hb{P}\) \\
	\ifftext either \(I_1(v), I_2(v) \neq \star\) or \(I_1(v) <_u I_2(v)\) for all \(v \in \hb{P}\) \\
	\ifftext \(I_1^{\textbf{u}} \subseteq I_2^{\textbf{u}}\).
\end{proofE}

Proposition~\ref{prop:NLP-u-order-subset} entails an equivalent definition for an L-stable model:

\begin{corollaryE}[][end, restate, text link=]\label{cor:NLP-L-StM-u-min-StPM}
	Given an \nlp \(P\), \(I\) is an L-stable model \ifftext it is a \(\leq_u\)-minimal stable partial model.
\end{corollaryE}
\begin{proofE}[text proof=]
	This immediately follows from Proposition~\ref{prop:NLP-u-order-subset}.
\end{proofE}

It is easy to derive that an L-stable model is also a regular model by Corollary~\ref{cor:NLP-L-StM-u-min-StPM} and the fact that \(I_1 \leq_s I_2\) implies \(I_1 \leq_u I_2\).
However, the reverse might be not true in general.
The relationship between \(\leq_s\) and \(\leq_u\) leads us to the following, which is similar to Theorem~\ref{theo:NLP-SuTS-reach-unique-SuPM}:

\begin{propositionE}[][end, restate, text link=]\label{prop:NLP-SuTS-reach-u-SuPM}
	Consider an \nlp \(P\).
	For every supported trap space \(I \in \suts{P}\), there is a supported partial model \(J \in \supms{P}\) such that \(J \leq_u I\).
\end{propositionE}
\begin{proofE}[text proof=]
	By Theorem~\ref{theo:NLP-SuTS-reach-unique-SuPM}, there is a unique supported partial model \(J \in \supms{P}\) such that \(J \leq_s I\).
	Since \(I' \leq_s I\) implies \(I' \leq_u I\), we can conclude the proof.
	Note that for \(\leq_u\), such a supported partial model may not be unique.
\end{proofE}

This property immediately implies that the \(\leq_u\)-minimal elements of the set of supported partial models and the set of supported trap spaces are the same.

\begin{corollaryE}[][end, restate, text link=]\label{cor:NLP-u-min-SuPM-u-min-SuTS}
	Consider an \nlp \(P\).
	A three-valued interpretation \(I\) is a \(\leq_u\)-minimal supported partial model of \(P\) \ifftext \(I\) is a \(\leq_u\)-minimal supported trap space of \(P\).
\end{corollaryE}
\begin{proofE}[text proof=]
	We have \(\supms{P} \subseteq \suts{P}\) by Proposition~\ref{prop:NLP-SuPM-is-SuTS}.
	
	``\(\Rightarrow\)'' Assume that \(I\) is not a \(\leq_u\)-minimal supported trap space of \(P\).
	Then there exists \(I' \in \suts{P}\) such that \(I' <_u I\).
	By Proposition~\ref{prop:NLP-SuTS-reach-u-SuPM}, there exists \(J \in \supms{P}\) such that \(J \leq_u I'\).
	This implies that \(J <_u I\), which contradicts the \(\leq_u\)-minimality of \(I\).
	Hence, \(I\) is a \(\leq_u\)-minimal supported trap space of \(P\).
	
	``\(\Leftarrow\)'' This direction is trivial because \(\supms{P} \subseteq \suts{P}\).
\end{proofE}

By using the least fixpoint transformation, we obtain the following equivalence:

\begin{theoremE}[][end, restate, text link=]\label{theo:NLP-L-StM-u-min-StTS}
	Given an \nlp \(P\), \(I\) is an L-stable model \ifftext \(I\) is a \(\leq_u\)-minimal stable trap space.
\end{theoremE}
\begin{proofE}[text proof=]
	We have \(P\) and \(\lfp{P}\) have the same set of stable partial models by Theorem~\ref{theo:NLP-lfp-model-equivalence}.
	They also have the same set of stable trap spaces by Corollary~\ref{cor:NLP-lfp-StTS}.
	Hence, it suffices to show that \(I\) is an L-stable model of \(\lfp{P}\) \ifftext \(I\) is a \(\leq_u\)-minimal stable trap space of \(\lfp{P}\).
	
	Recall that \(\lfp{P}\) is a negative NLP.
	We have \(I\) is an L-stable model of \(\lfp{P}\) \\
	\ifftext \(I\) is a \(\leq_u\)-minimal stable partial model of \(\lfp{P}\) by Corollary~\ref{cor:NLP-L-StM-u-min-StPM} \\
	\ifftext \(I\) is a \(\leq_u\)-minimal supported partial model of \(\lfp{P}\) by Theorem~\ref{theo:neg-NLP-StPM-SuPM} \\
	\ifftext \(I\) is a \(\leq_u\) supported trap space of \(\lfp{P}\) by Corollary~\ref{cor:NLP-u-min-SuPM-u-min-SuTS} \\
	\ifftext \(I\) is a \(\leq_u\)-minimal stable trap space of \(\lfp{P}\) by Corollary~\ref{cor:neg-NLP-StTS-SuTS}.
\end{proofE}

Combining the above relationship, the least fixpoint transformation, and Theorem~\ref{theo:NLP-RegM-s-min-StTS}, we obtain the following:

\begin{propositionE}[][end, restate, text link=]\label{prop:NLP-L-StM-u-min-RegM}
	Given an \nlp \(P\), \(I\) is an L-stable model of \(P\) \ifftext \(I\) is a \(\leq_u\)-minimal regular model.
\end{propositionE}
\begin{proofE}[text proof=]
	By Theorem~\ref{theo:NLP-L-StM-u-min-StTS}, it suffices to show that \(I\) is a \(\leq_u\)-minimal stable trap space of \(P\) \ifftext \(I\) is a \(\leq_u\)-minimal regular model of \(P\).
	The forward direction is trivial because \(\regms{P} \subseteq \stts{P}\), i.e., the set of regular models is a subset of the set of stable trap spaces.
	Now, we prove the backward direction.
	
	Assume that \(I\) is a \(\leq_u\)-minimal regular model of \(P\).
	Then \(I\) is a \(\leq_u\)-minimal regular model of \(\lfp{P}\) by Theorem~\ref{theo:NLP-lfp-model-equivalence}.
	Suppose that \(I\) is not a \(\leq_u\)-minimal stable trap space of \(\lfp{P}\).
	Then there is a stable trap space \(I'\) of \(\lfp{P}\) such that \(I' <_u I\).
	By Theorem~\ref{theo:NLP-RegM-s-min-StTS}, the regular models of \(\lfp{P}\) are the \(\leq_s\)-minimal stable trap spaces of \(\lfp{P}\).
	This implies that there is a regular model \(J\) of \(\lfp{P}\) such that \(J \leq_s I'\).
	Since \(J \leq_s I'\) implies \(J \leq_u I'\), \(J <_u I\).
	By Theorem~\ref{theo:NLP-lfp-model-equivalence}, \(J\) is a regular model of \(P\), which contradicts the \(\leq_u\)-minimality of \(I\).
	Hence, \(I\) is a \(\leq_u\)-minimal stable trap space of \(\lfp{P}\).
	By Corollary~\ref{cor:NLP-lfp-StTS}, \(I\) is a \(\leq_u\)-minimal stable trap space of \(P\).
\end{proofE}

The work by~\cite{SZ1997} shows that an L-stable model is a regular model.
Proposition~\ref{prop:NLP-L-StM-u-min-RegM} provides a deeper relationship between the L-stable models and the regular models.
Interestingly, this helps us to reveal the following insight:

\begin{corollaryE}[][end, restate, text link=]\label{cor:NLP-num-L-StM-RegM}
	If an \nlp \(P\) has no L-stable model, then it has infinitely many regular models.
\end{corollaryE}
\begin{proofE}[text proof=]
	By Proposition~\ref{prop:NLP-L-StM-u-min-RegM}, the set of L-stable models of \(P\) coincides with the set of \(\leq_u\)-minimal regular models of \(P\).
	If \(\regms{P}\) is a finite set, then it always has a \(\leq_u\)-minimal element, i.e., the set of \(\leq_u\)-minimal regular models of \(P\) is non-empty, leading to \(P\) has an L-stable model, which is a contradiction.
	Hence, \(P\) has infinitely many regular models.
\end{proofE}

%% file: sections/consistency.tex
\section{Existence Results}\label{sec:semantics-consistency}

Whether a given semantics is consistent (i.e., it contains at least one respective object) is a main property of interest~\cite{Sato1990,Spanring2017}.
For \pnames, proving or refuting this property is mostly straightforward~\cite{BS2014}.
However, for arbitrary \nlps with an infinite Herbrand base, conducting such proofs or refutations is more intricate.
It often needs a precise use of set-theoretic axioms or equivalent statements.
We here show that previous arguments for the existence of supported classes, strict supported (stable) classes, and regular models are either missing or imprecise; and that the trap space semantics can offer a unified and precise framework to prove the model or class existence.
Furthermore, we reveal problems in the proof of~\cite{SZ1997} for the existence of L-stable models and correct it for uni-rule programs only (see more details in Section~\ref{sec:existence-L-stable}).


We first recall a well-known result:

\begin{theorem}[Theorem 1 of~\cite{BS1992}]\label{theo:NLP-exist-StC}
	Any NLP \(P\) has a stable class.
\end{theorem}

The proof of Theorem~\ref{theo:NLP-exist-StC} relies on the anti-monotonicity of \(F_P\) and the monotonicity of \(F_P^2\) to construct two two-valued interpretations that form a stable class.
However, this approach may not be directly applicable to supported classes, as \(T_P\) may be neither monotonic nor anti-monotonic.
In~\cite[Proposition 5.6]{IS2012}, the authors assert the guaranteed existence of stable and supported classes but do not provide formal proofs.
While they state that ``As explained in previous sections, existence of stable/supported classes is always guaranteed'', our analysis indicates that this claim is not clearly substantiated by the preceding discussions in their work.
Relying on the existence of \(\subseteq\)-minimal stable (resp.\ supported) trap sets, which can be easily derived from Lemma~\ref{lem:NLP-St-trap-set-min-s} (resp.\ Lemma~\ref{lem:NLP-Su-trap-set-min-s}), we provide an alternative proof for the existence of stable classes (resp.\ a new proof for the existence of supported classes) in an \nlp.

\begin{theoremE}[][end, restate, text link=]\label{theo:NLP-exist-min-trap-set}
	An \nlp \(P\) always has a \(\subseteq\)-minimal stable (or supported) trap set.
\end{theoremE}
\begin{proofE}[text proof=]
	For the case of \(\subseteq\)-minimal stable trap sets, this immediately follows from Lemma~\ref{lem:NLP-St-trap-set-min-s} and the fact that \(P\) has at least one stable trap set by Proposition~\ref{prop:NLP-exist-St-Su-TS}.
	
	For the case of \(\subseteq\)-minimal supported trap sets, this immediately follows from Lemma~\ref{lem:NLP-Su-trap-set-min-s} and the fact that \(P\) has at least one supported trap set by Proposition~\ref{prop:NLP-exist-St-Su-TS}.
\end{proofE}

\begin{proofE}[text proof=Alternative proof of Theorem 5.1]
	By Theorem~\ref{theo:NLP-exist-min-trap-set}, \(P\) has a \(\subseteq\)-minimal stable trap set \(S\).
	We show that \(S\) is a stable class.
	
	By definition, \(\{\Fop{I} \mid I \in S\} \subseteq S\).
	Suppose that \(\{\Fop{I} \mid I \in S\} \subsetneq S\).
	Then there exists \(J \in S\) such that for all \(I \in S\), \(\Fop{I} \neq J\).
	We have \(\{\Fop{I} \mid I \in S \setminus \{J\}\} \subseteq \{\Fop{I} \mid I \in S\} \subseteq S \setminus \{J\}\).
	Thus \(S \setminus \{J\}\) is a supported trap set of \(P\), which contradicts the \(\subseteq\)-minimality of \(S\).
	Hence \(\{\Fop{I} \mid I \in S\} = S\), i.e., \(S\) is a supported class of \(P\).
\end{proofE}

\begin{theoremE}[][end, restate, text link=]\label{theo:NLP-exist-SuC}
	Any \nlp \(P\) has a supported class.
\end{theoremE}
\begin{proofE}[text proof=]
	By Theorem~\ref{theo:NLP-exist-min-trap-set}, \(P\) has a \(\subseteq\)-minimal supported trap set \(S\).
	We show that \(S\) is a supported class.
	
	By definition, \(\{\Top{I} \mid I \in S\} \subseteq S\).
	Suppose that \(\{\Top{I} \mid I \in S\} \subsetneq S\).
	Then there exists \(J \in S\) such that for all \(I \in S\), \(\Top{I} \neq J\).
	We have \(\{\Top{I} \mid I \in S \setminus \{J\}\} \subseteq \{\Top{I} \mid I \in S\} \subseteq S \setminus \{J\}\).
	Thus \(S \setminus \{J\}\) is a supported trap set of \(P\), which contradicts the \(\subseteq\)-minimality of \(S\).
	Hence \(\{\Top{I} \mid I \in S\} = S\), i.e., \(S\) is a supported class of \(P\).
\end{proofE}

While the existence of stable classes is precisely proved in~\cite{BS1992}, there has been no formal proof for the existence of strict stable (supported) classes so far.
The work by~\cite{BS1993} considers (finite) stable classes and suggests that verifying the existence of strict stable classes in such cases is straightforward, but even under these constraints, the derivation is not always evident indeed.
Thanks to the relationships with \(\subseteq\)-minimal stable (resp.\ supported) trap sets established in the previous section, we can easily derive the existence of strict stable (resp.\ supported) classes:

\begin{corollaryE}[\textbf{main result}][end, restate, text link=]\label{cor:NLP-exist-strict-StC-SuC}
	An \nlp \(P\) always has a strict stable (or supported) class.
\end{corollaryE}
\begin{proofE}[text proof=]
	This immediately follows from Theorem~\ref{theo:NLP-exist-min-trap-set} and Theorem~\ref{theo:NLP-min-St-trap-set-strict-StC} (resp.\ Theorem~\ref{theo:NLP-min-Su-trap-set-strict-SuC}).
\end{proofE}


The equivalence between the regular models and the \(\leq_s\)-minimal stable trap spaces (see Theorem~\ref{theo:NLP-RegM-s-min-StTS}) and the existence of \(\leq_s\)-minimal stable trap spaces (see Theorem~\ref{theo:NLP-exist-min-s-StTS-SuTS}) immediately entail the existence of regular models in an \nlp:

\begin{corollaryE}[][end, restate, text link=]\label{cor:NLP-RegM-exist}
	An \nlp \(P\) always has a regular model.
\end{corollaryE}
\begin{proofE}[text proof=]
	This immediately follows from Theorem~\ref{theo:NLP-RegM-s-min-StTS} and Theorem~\ref{theo:NLP-exist-min-s-StTS-SuTS}.
\end{proofE}

\cite[Proposition 5.1]{YY1994} proves the existence of regular models in an \nlp by showing that the \nlp has at least one justifiable model (i.e., stable partial model) and that \(<_{undef}\) (i.e.,\(<_s\)) is a quasi-order (i.e., a transitive and reflexive relation) on the set of justifiable models.
However, this argument is only valid for finite sets, whereas the set of stable partial models of an \nlp can be infinite.
An infinite set can have no minimal element in general (see more examples in~\cite{Spanring2017}).
For infinite sets, it should follow the axiom of choice or equivalent statements~\cite{BS2014}.
With our new semantics, we can have an elegant proof for the existence of regular models that is precise w.r.t. the use of set-theoretic axioms.

Note that for \nlps, regular models are equivalent to M-stable models~\cite{ELS1997}.
Since the existence of M-stable models is proved in~\cite[Corollary 5.4]{SZ1997}, the existence of regular models immediately follows.
Hence, our proof of Corollary~\ref{cor:NLP-RegM-exist} offers an alternative (perhaps simpler) for the proof of~\cite[Corollary 5.4]{SZ1997}.

%% file: sections/conclusion.tex
\section{Conclusion}\label{sec:conclusion}

In this paper, we introduce the \emph{trap space semantics} as a novel and general framework for interpreting \nlps. 
Extending prior work on \pnames and their correspondence with Boolean networks, we generalize the concept of trap spaces to apply to arbitrary \nlps, including those with infinite Herbrand bases.
This generalization provides a unified, semantics-based foundation that captures both model-theoretic and dynamical perspectives on \nlps.

We formally define the trap space semantics and established its foundational properties.
Notably, we demonstrate that it encompasses and connects existing model-theoretic and dynamical semantics.
Through this unification, the trap space semantics clarifies the relationships among these diverse approaches and provides a systematic means to compare and interpret them.
See Figure~\ref{fig:summary-relationships} for a summary of the main relationships shown in Section~\ref{sec:TS-relationships}.

\begin{figure}[!ht]
\centering
\begin{tikzpicture}[
	node distance=1.5cm and 1cm,
	every node/.style={draw, rounded corners=5pt, font=\small, align=center, inner sep=6pt, black},
	blackarrow/.style={->, thick, color=black},
	]
	
	\node (stableclass) {stable\\class};
	\node[right=of stableclass] (stablemodel) {stable\\model};
	\node[right=of stablemodel] (lstable) {L-stable\\model};
	\node[right=of lstable] (regular) {regular\\model};
	\node[right=of regular] (partial) {stable\\partial model};
	
	\node[below=of stablemodel,yshift=-1cm] (supportedmodel) {supported\\ model};
	\node[left=of supportedmodel] (supportedclass) {supported\\class};
	\node[below=of partial,yshift=-1cm] (spmodel) {supported\\partial model};
	\node[right=of partial] (StTS) {stable\\trap space};
	\node[below=of StTS,yshift=-1cm] (SuTS) {supported\\trap space};
	
	\draw[blackarrow] (stablemodel) -- (stableclass);
	\draw[blackarrow] (stablemodel) -- (lstable);
	\draw[blackarrow] (lstable) -- (regular);
	\draw[blackarrow] (regular) -- (partial);
	\draw[blackarrow] (supportedmodel) -- (supportedclass);
	\draw[blackarrow] (supportedmodel) -- (spmodel);
	\draw[blackarrow] (spmodel) -- (partial);
	\draw[blackarrow] (supportedmodel.north) -- (stablemodel.south);
	\draw[blackarrow] (partial) -- (StTS);
	\draw[blackarrow] (spmodel) -- (SuTS);
	
	\draw[] (stableclass) edge [bend left=45,dashed] node [midway, above, fill=white] {Corollary~\ref{cor:NLP-min-TS-contain-strict-class}: each stable trap space contains at least on strict stable class} (StTS);
	
	\draw[] (supportedclass) edge [bend right=45,dashed] node [midway, below, fill=white] {Corollary~\ref{cor:NLP-min-TS-contain-strict-class}: each supported trap space contains at least on strict supported class} (SuTS);
	
	\draw[] (regular) edge [bend left=30,dashed] node [midway, above, fill=white] {Theorem~\ref{theo:NLP-RegM-s-min-StTS}: regular models = \(\leq_s\)-minimal stable trap spaces} (StTS);
	
	\draw[] (spmodel) edge [bend right=30,dashed] node [midway, below, fill=white,xshift=-1cm] {Corollary~\ref{cor:NLP-min-s-SuPM-min-s-SuTS}: \(\leq_s\)-minimal supported partial models \\= \(\leq_s\)-minimal supported trap spaces} (SuTS);
	
	\draw[] (lstable) edge [bend right=30,dashed] node [midway, below, fill=white,xshift=-1cm] {Theorem~\ref{theo:NLP-L-StM-u-min-StTS}: L-stable models \\= \(\leq_u\)-minimal supported trap spaces} (StTS);
	
\end{tikzpicture}
\caption{Summary of the main relationships shown in Section~\ref{sec:TS-relationships}. A full arc indicates a set-inclusion relation, whereas a dashed line along with a box indicate the relationship between the two end vertices.}
\label{fig:summary-relationships}
\end{figure}
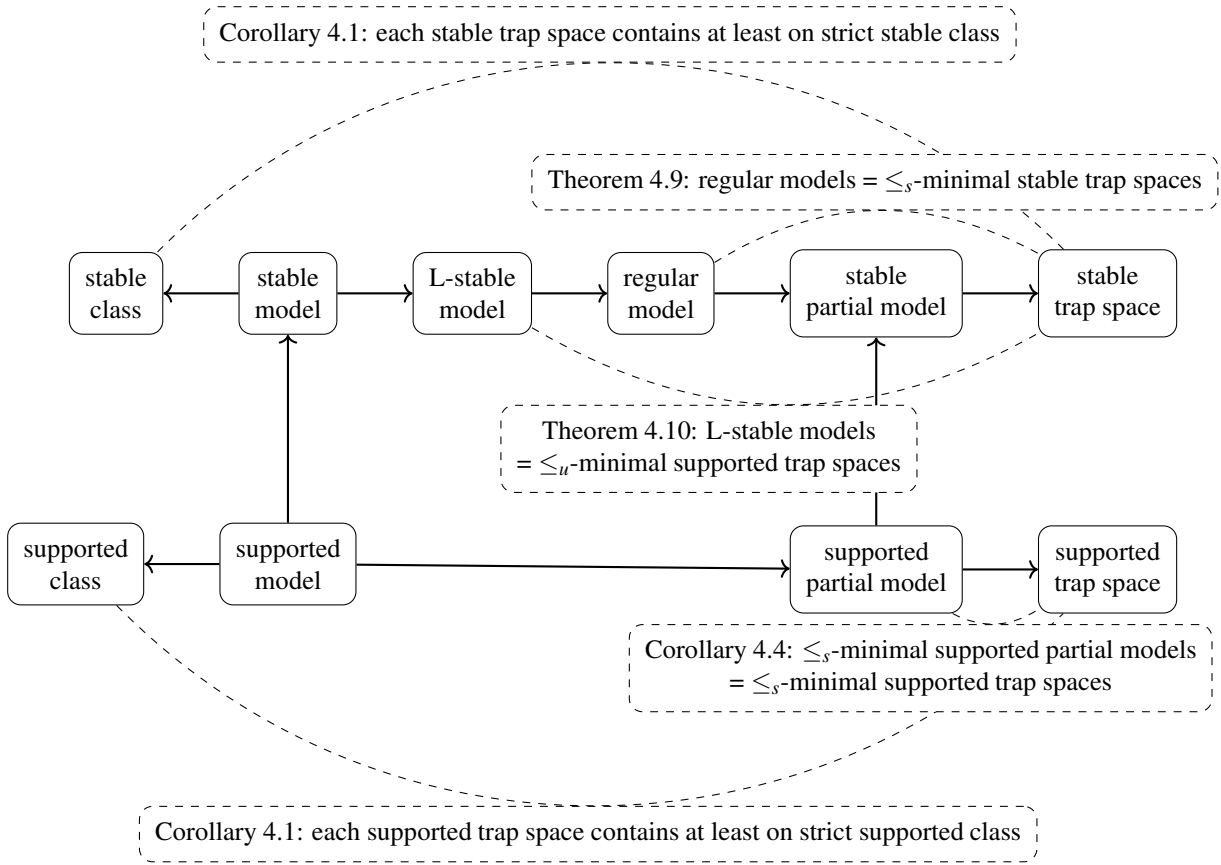

Another significant contribution is the use of the trap space semantics as a unified framework for establishing existence results.
We demonstrate how this framework enables rigorous proofs of the existence of supported classes, strict supported classes, strict stable classes, and regular models—results which were previously either absent or only informally stated in the literature.

More broadly, the trap space semantics advances the field by offering a conceptual bridge between the model-theoretic and dynamical views of \nlps. 
It allows us to reason simultaneously about fixed points of semantics-based operators and the dynamical behavior of \nlps over time.
This duality opens up new avenues for understanding the behavior of \nlps comprehensively.

Looking forward, future research should focus on extending the trap space semantics to richer classes of programs, such as disjunctive logic programs, constraint logic programs, and extended logic programs.
Furthermore, there are many other semantic proposals, for example, those that consider abstract properties such as cumulativity~\cite{Dix1995a,Dix1995b}.
Relating the trap space semantics to such semantics is interesting and can help clarify the semantic landscape.
Finally, it is worth exploring the semantics' applications in areas that benefit from the interplay between logic and dynamics, including knowledge representation, verification, and systems biology.

%% file: sections/detailed-proofs.tex
\section{Detailed Proofs}\label{sec:appendix-detailed-proofs}

This appendix shows the detailed proofs of the results shown in the main text of the present paper.

\printProofs

%% file: sections/characterization-st-su-class.tex
\section{Revisit of the Characterizations of Strict Stable and Supported Classes}\label{sec:char-st-su-class}

\renewcommand{\thetheorem}{\Alph{section}.\arabic{theorem}}
\renewcommand{\thedefinition}{\Alph{section}.\arabic{definition}}
\renewcommand{\theproposition}{\Alph{section}.\arabic{proposition}}
\renewcommand{\thelemma}{\Alph{section}.\arabic{lemma}}

We first revisit the basic definitions and results of~\cite{IS2012} on stable and supported classes.

\begin{definition}[\cite{IS2012}]\label{def:NLP-orbit}
	Consider an \nlp \(P\) and a two-valued interpretation \(I\) of \(P\).
	A sequence of applications of an operator on the set of two-valued interpretations of \(P\) is called an \emph{orbit}.
	The orbit of \(I\) w.r.t.\ \(T_P\) is the sequence \(\langle T_P^k(I)\rangle_{k \in \omega}\), where \(T_P^0(I) = I\) and \(T_P^{k + 1}(I) = T_P(T_P^{k}(I))\).
	Similarly, we have the definition for the orbit of \(I\) w.r.t.\ \(F_P\).
\end{definition}

\begin{definition}[\cite{IS2012}]\label{def:NLP-orbit-elements}
	Consider an \nlp \(P\) and a two-valued interpretation \(I\) of \(P\).
	Let \(\mathcal{T}_{P}(I)\) be the set of two-valued interpretations defined as: 
		\[
			\mathcal{T}_{P}(I) = \{T_P^k(I) \mid k \in \omega\}.
		\]
	Let \(\mathcal{F}_{P}(I)\) be the set of two-valued interpretations defined as: 
		\[
			\mathcal{F}_{P}(I) = \{F_P^k(I) \mid k \in \omega\}.
		\]
\end{definition}

\begin{theorem}[\cite{IS2012}]\label{theo:NLP-strict-StC-SuC-orbit}
	Given an \nlp \(P\), a non-empty set \(S\) of two-valued interpretations is a strict stable (resp.\ supported) class of \(P\) \ifftext \(\mathcal{T}_{P}(I) = S\) (resp.\ \(\mathcal{F}_{P}(I) = S\)) for every \(I \in S\).
\end{theorem}

Theorem~\ref{theo:NLP-strict-StC-SuC-orbit} provides an alternative characterization for strict stable (or supported) classes. 
This theorem mainly relies on the definitions of stable and supported orbits (see Definition~\ref{def:NLP-orbit}). 
This may be a problem here, as these definitions only limit the length of the sequence to \(\omega\) (the first infinite limit ordinal).
Since the sequence can traverse through the set of all two-valued interpretations, which may be an infinite set, thus the length can be a bigger ordinal.
With Definition~\ref{def:NLP-orbit}, it seems that Theorem~\ref{theo:NLP-strict-StC-SuC-orbit} only consider \pnames rather than any \nlps.
We fix the above problem as follows.

First, we recursively define the stable or supported reachable closure of a two-valued interpretation over all ordinals.

\begin{definition}\label{def:NLP-Su-St-reachable-closure}
	Given an \nlp \(P\), let \(I\) is a two-valued interpretation of \(P\).
	The \emph{supported reachable closure} of \(I\) is defined as: 
			\begin{align*}
					\Pi^{I}_{P} \uparrow 0 &= \{I\} \\
					\Pi^{I}_{P} \uparrow \alpha + 1 &= \Pi^{I}_{P} \uparrow \alpha \cup \bigcup_{J \in \Pi^{I}_{P} \uparrow \alpha}\{\Top{J}\}.
				\end{align*}
	If \(\beta\) is a limit ordinal, then by construction \(\{\Pi^{I}_{P} \uparrow \alpha \colon \alpha < \beta\}\) is a \(\subseteq\) chain.
	Define \(\Pi^{I}_{P} \uparrow \beta = \text{sup}_{\subseteq}\{\Pi^{I}_{P} \uparrow \alpha \colon \alpha < \beta\}\), where \(\text{sup}_{\subseteq}\) denotes the least upper bound w.r.t. \(\subseteq\).
	By replacing \(\Top{J}\) and \(\Pi^{I}_{P}\) with \(\Fop{J}\) and \(\Sigma^{I}_{P}\) respectively, we obtain the definition of the \emph{stable reachable closure} of \(I\).
\end{definition}


Then the two following propositions show that the stable or supported reachable closure eventually reaches a constant \wrttext any two-valued interpretation.
The key to their proofs is the use of Hartogs' lemma (see Lemma~\ref{lem:Hartogs}).

\begin{proposition}\label{prop:NLP-Su-reachable-closure-constant}
	Consider an \nlp \(P\).
	For any two-valued interpretation \(I\) of \(P\), there is the unique least ordinal \(\alpha_{I}\) such that \(\Pi^{I}_{P} \uparrow \alpha_{I} + 1 = \Pi^{I}_{P} \uparrow \alpha_{I}\).
\end{proposition}
\begin{proof}
	We have \(\Pi^{I}_{P} \uparrow\) is an increasing function from the ordinals into \(2^{\hb{P}}\).
	It cannot be strictly increasing, as if it was, then we would have an injection from the ordinals into \(2^{\hb{P}}\), which violates Lemma~\ref{lem:Hartogs}.
	Hence, \(\Pi^{I}_{P} \uparrow\) must be eventually constant.
	It follows that for some ordinal \(\alpha\), \(\Pi^{I}_{P} \uparrow \alpha + 1 = \Pi^{I}_{P} \uparrow \alpha\).
	Let \(\alpha_{I}\) be the least ordinal such that \(\Pi^{I}_{P} \uparrow \alpha_{I} + 1 = \Pi^{I}_{P} \uparrow \alpha_{I}\).
	Since the ordinals are ordered, \(\alpha_{I}\) is unique.
\end{proof}

\begin{proposition}\label{prop:NLP-St-reachable-closure-constant}
	Consider an \nlp \(P\).
	For any two-valued interpretation \(I\) of \(P\), there is the unique least ordinal \(\alpha_{I}\) such that \(\Sigma^{I}_{P} \uparrow \alpha_{I} + 1 = \Sigma^{I}_{P} \uparrow \alpha_{I}\).
\end{proposition}
\begin{proof}
	Replace \(\Pi^{I}_{P}\) with \(\Sigma^{I}_{P}\) in the proof of Proposition~\ref{prop:NLP-Su-reachable-closure-constant}.
\end{proof}

By using transfinite induction (see more details in~\cite{Lloyd1984}), we prove that the supported reachable closure is closed under the supported trap set property.

\begin{lemma}\label{lem:NLP-Su-trap-set-closure-closed}
	Consider an \nlp \(P\).
	Let \(S\) be a supported trap set of \(P\).
	Then for all \(I \in S\), \(\Pi^{I}_{P} \uparrow \alpha \subseteq S\) for all ordinal \(\alpha\).
\end{lemma}
\begin{proof}
	We prove this by using transfinite induction.
	
	Base case: \(\Pi^{I}_{P} \uparrow 0 = \{I\} \subseteq S\).
	
	Successor case: \(\Pi^{I}_{P} \uparrow \alpha + 1  = \Pi^{I}_{P} \uparrow \alpha \cup \bigcup_{J \in \Pi^{I}_{P} \uparrow \alpha}\{\Top{J}\} \subseteq S \cup \bigcup_{J \in S}\{\Top{J}\}\) because \(\Pi^{I}_{P} \uparrow \alpha \subseteq S\) by the induction hypothesis.
	Then \(\Pi^{I}_{P} \uparrow \alpha + 1 \subseteq S \cup S = S\) because \(S\) is a supported trap set.
	
	Limit case: \(\Pi^{I}_{P} \uparrow \alpha = \text{sup}_{\subseteq}\{\Pi^{I}_{P} \uparrow \beta \colon \beta < \alpha\} \subseteq S\) because \(\Pi^{I}_{P} \uparrow \beta \subseteq S\) for all \(\beta < \alpha\) by the induction hypothesis.
\end{proof}

On the other hand, we prove that the constant that the supported reachable closure reaches is a supported trap set \wrttext any two-valued interpretation.

\begin{lemma}\label{lem:NLP-closure-is-Su-trap-set}
	Given an \nlp \(P\), for any two-valued interpretation \(I\) of \(P\), \(\Pi^{I}_{P} \uparrow \alpha_{I}\) is a supported trap set of \(P\).
\end{lemma}
\begin{proof}
	By definition of \(\Pi^{I}_{P}\), we have \(\Pi^{I}_{P} \uparrow \alpha_{I} + 1 = \Pi^{I}_{P} \uparrow \alpha_{I} \cup \bigcup_{J \in \Pi^{I}_{P} \uparrow \alpha_{I}}\{\Top{J}\}\).
	By definition of \(\alpha_{I}\), we have \(\Pi^{I}_{P} \uparrow \alpha_{I} + 1 = \Pi^{I}_{P} \uparrow \alpha_{I}\).
	Set \(S = \Pi^{I}_{P} \uparrow \alpha_{I}\).
	It follows that \(S = S \cup \bigcup_{J \in S}\{\Top{J}\}\).
	Then \(\bigcup_{J \in S}\{\Top{J}\} \subseteq S\), implying that \(S\) is a supported trap set of \(P\).
\end{proof}

Combining Lemma~\ref{lem:NLP-Su-trap-set-closure-closed} and Lemma~\ref{lem:NLP-closure-is-Su-trap-set}, we now obtain a new characterization for \(\subseteq\)-minimal supported trap sets of an \nlp:

\begin{theorem}\label{theo:NLP-min-Su-trap-set-Su-closure}
	Given an \nlp \(P\), \(S\) is a \(\subseteq\)-minimal supported trap set of \(P\) \ifftext \(\Pi^{I}_{P} \uparrow \alpha_{I} = S\) for all \(I \in S\).
\end{theorem}
\begin{proof}
	``\(\Rightarrow\)'' Assume that \(S\) is a \(\subseteq\)-minimal supported trap set of \(P\).
	Suppose that there exists \(I \in S\) such that \(\Pi^{I}_{P} \uparrow \alpha_{I} \neq S\).
	By Lemma~\ref{lem:NLP-Su-trap-set-closure-closed}, \(\Pi^{I}_{P} \uparrow \alpha_{I} \subseteq S\) since \(S\) is a supported trap set.
	Then \(\Pi^{I}_{P} \uparrow \alpha_{I} \subset S\).
	By Lemma~\ref{lem:NLP-closure-is-Su-trap-set}, \(\Pi^{I}_{P} \uparrow \alpha_{I}\) is a supported trap set, which contradicts the \(\subseteq\)-minimality of \(S\).
	Hence, \(\Pi^{I}_{P} \uparrow \alpha_{I} = S\) for all \(I \in S\).
	
	``\(\Leftarrow\)'' Assume that \(\Pi^{I}_{P} \uparrow \alpha_{I} = S\) for all \(I \in S\).
	We have \(\Pi^{I}_{P} \uparrow \alpha_{I} + 1 = S = S \cup \bigcup_{J \in S}\{\Top{J}\}\).
	This implies that \(\bigcup_{J \in S}\{\Top{J}\} \subseteq S\), thus \(S\) is a supported trap set of \(P\).
	Suppose that \(S\) is not \(\subseteq\)-minimal.
	Then there is a supported trap set \(S'\) of \(P\) such that \(S' \subset S\).
	Consider \(J \in S'\).
	By Lemma~\ref{lem:NLP-Su-trap-set-closure-closed}, \(\Pi^{I}_{P} \uparrow \alpha_{J} \subseteq S' \subset S\), which is a contradiction because \(J \in S\).
	Hence, \(S\) is a \(\subseteq\)-minimal supported trap set of \(P\).
\end{proof}

This is similar for stable trap sets by replacing \(\Pi^{I}_{P}\) by \(\Sigma^{I}_{P}\) in the proof of Theorem~\ref{theo:NLP-min-Su-trap-set-Su-closure}:

\begin{theorem}\label{theo:NLP-min-St-trap-set-St-closure}
	Given an \nlp \(P\), \(S\) is a \(\subseteq\)-minimal stable trap set of \(P\) \ifftext \(\Sigma^{I}_{P} \uparrow \alpha_{I} = S\) for all \(I \in S\).
\end{theorem}
\begin{proof}
	Replace \(\Pi^{I}_{P}\) by \(\Sigma^{I}_{P}\) in the proof of Theorem~\ref{theo:NLP-min-Su-trap-set-Su-closure}.
\end{proof}

Theorem~\ref{theo:NLP-min-Su-trap-set-Su-closure} and Theorem~\ref{theo:NLP-min-St-trap-set-St-closure} are the corrected version of Theorem~\ref{theo:NLP-strict-StC-SuC-orbit} in the case of \nlps, as we show that the \(\subseteq\)-minimal supported (resp.\ stable) trap sets are equivalent to strict supported (resp.\ stable) classes (see Theorem~\ref{theo:NLP-min-St-trap-set-strict-StC} and Theorem~\ref{theo:NLP-min-Su-trap-set-strict-SuC}).
In the case of \pnames, \(\hb{P}\) is finite, then \(\mathcal{T}_{P}(I) = \Pi^{I}_{P} \uparrow \omega\) (resp.\ \(\mathcal{F}_{P}(I) = \Sigma^{I}_{P} \uparrow \omega\)). 
Hence, in this case, Theorem~\ref{theo:NLP-min-Su-trap-set-Su-closure} and Theorem~\ref{theo:NLP-min-St-trap-set-St-closure} are equivalent to Theorem~\ref{theo:NLP-strict-StC-SuC-orbit}.

%% file: sections/existence-L-stable.tex
\section{Revisit of the Existence of L-stable Models}\label{sec:existence-L-stable}

\renewcommand{\thetheorem}{\Alph{section}.\arabic{theorem}}
\renewcommand{\theexample}{\Alph{section}.\arabic{example}}
\renewcommand{\theconjecture}{\Alph{section}.\arabic{conjecture}}

The existence of L-stable models in an \nlp is stated and proved in~\cite[Proposition 5.11]{SZ1997}.
The main idea of the proof is similar to what is used for proving the existence of M-stable models in an \nlp (see~\cite[Theorem 5.3]{SZ1997}).
However, this may pose a problem here because the case of L-stable models is actually not the same as the case of M-stable models.
More specifically, the union of all the partial stable models in a maximal chain w.r.t. \(\subseteq\)-minimal \(I^{\textbf{u}}\) might be not a partial stable model (even the union does not exist, see Example~\ref{exam:Datalog-non-exist-union-maximal-chain-u}), whereas the union of all the stable partial models in a \(\leq_i\)-chain is a stable partial model as shown in the proof of~\cite[Theorem 5.3]{SZ1997}.
This problem is more justified by the bijection between the L-stable models of an \nlp and the semi-stable extensions of an Abstract Argumentation Framework (\af)~\cite{CSAD2015}.
It has been shown that there is an infinite \af without a semi-stable extension~\cite{BS2014}.

\begin{example}\label{exam:Datalog-non-exist-union-maximal-chain-u}
	Consider the \nlp \(P = \{a \leftarrow \dng{b}; b \leftarrow \dng{a}; c \leftarrow \dng{c}\}\).
	The program \(P\) has three stable partial models: \(I_1 = \{a = \uval, b = \uval, c = \uval\}\), \(I_2 = \{a = 0, b = 1, c = \uval\}\), and \(I_3 = \{a = 1, b = 0, c = \uval\}\).
	Using the notations of~\cite{SZ1997}, we have \(I_1 = \langle \emptyset, \emptyset\rangle\), \(I_1 = \langle\{b\}, \{a\}\rangle\), and \(I_1 = \langle\{a\}, \{b\}\rangle\).
	Every maximal chain w.r.t. \(\subseteq\)-minimal \(I^{\textbf{u}}\) of \(P\) contains all \(I_1\), \(I_2\), and \(I_3\).
	However, the union of all elements in such a chain is \(\langle\{a, b\}, \{a, b\}\rangle\), which does not represent a valid three-valued interpretation.
	Hence, the union does not exist.
\end{example}

We make the following two conjectures, inspired by existing related results in abstract argumentation frameworks~\cite{BS2014,Spanring2017}.

\begin{conjecture}\label{conj:NLP-internal-variable-free-L-StM-exist}
	An internal-variable-free \nlp \(P\) has at least one L-stable model.
\end{conjecture}

\begin{conjecture}\label{conj:NLP-L-StM-non-exist}
	There exists an \nlp \(P\) such that \(P\) has no L-stable model.
\end{conjecture}

We here prove a special case of Conjecture~\ref{conj:NLP-internal-variable-free-L-StM-exist}, however this conjecture and Conjecture~\ref{conj:NLP-L-StM-non-exist} are still open.
This result is meaningful because 1) the class of uni-rule \nlps is easily identifiable by a syntatic characterization, 2) uni-rule \pnames are as hard as general \pnames w.r.t the stable model semantics, since identifying whether a uni-rule \pname has a stable model or not is NP-complete~\cite{SS1997}, and 3) this class is straightforwardly related to (possibly infinite) \afs~\cite{Dung1995,CG2009}.

\begin{theorem}\label{theo:uni-rule-NLP-exist-L-StM}
	Consider an \nlp \(P\).
	If \(P\) is uni-rule, then it has at least one L-stable model.
\end{theorem}
\begin{proof}
	Since \(P\) is uni-rule, it is easy to derive that \(\lfp{P}\) is also uni-rule thanks to the definition of the least fixpoint.
	If \(\lfp{P} = \emptyset\), then it trivially has a L-stable model \(I^{*}\) where \(I^{*}(a) = \fval\) for every \(a \in \hb{\lfp{P}}\).
	We then only consider the case \(\lfp{P} \neq \emptyset\).
	If there is an atom \(a \in \hb{\lfp{P}}\) not being the head of any rule in \(\lfp{P}\), it is always \(\fval\) in any stable or supported partial model of \(\lfp{P}\).
	Hence we can ``propagate'' such atoms to get a new \nlp \(P'\) such that the set of stable partial models of \(P'\) one-to-one corresponds to the set of stable partial models of \(\lfp{P}\).
	This implies that the set of L-stable models of \(P'\) one-to-one corresponds to the set of L-stable models of \(\lfp{P}\).
	Since \(\lfp{P}\) is uni-rule and negative, \(P'\) also does.
	There is a (possibly infinite) \af (say \(\mathcal{A}\)) corresponding to \(P'\) where the set of arguments of \(\mathcal{A}\) is \(\hb{P'}\) and each atom in the body of a rule whose head is \(a\) corresponds to an attacker of \(a\) in \(\mathcal{A}\).
	Since the body of a rule is finite by definition, \(\mathcal{A}\) is finitary, i.e., every argument has finitely many attackers~\cite{BS2014}.
	
	A labelling of an \af is identical to a three-valued interpretation of an \nlp.
	An extension of an \af is defined as a subset of arguments and correponds to a two-valued interpretation of an \nlp.
	There is a one-to-one correspondence between the set of labellings and the set of extensions of a given type~\cite{CG2009}.
	\cite[Theorem 14]{BS2014} points out that a finitary \af has at least one semi-stable extension.
	We can easily to derive the existence of semi-stable labellings in \(\mathcal{A}\).
	
	A complete labelling \(I\) is a semi-stable labelling if \(I^{\textbf{u}}\) is \(\subseteq\)-minimal~\cite{CSAD2015}.
	By the definition of complete labellings~\cite{CG2009}, we can derive that the set of complete lablellings of \(\mathcal{A}\) coincides with the set of supported partial models of \(P'\).
	Since the stable partial models of \(P'\) coincide with the supported partial models of \(P'\) by Theorem~\ref{theo:neg-NLP-StPM-SuPM}, \(P'\) has at least one L-stable model, leading to \(\lfp{P}\) does.
	Since \(P\) and \(\lfp{P}\) have the same set of stable partial models by Theorem~\ref{theo:NLP-lfp-model-equivalence}, \(P\) has at least one L-stable model.
\end{proof}